\documentclass[11pt,a4paper]{article}
\usepackage{fullpage,amssymb,amsfonts,amsthm,paralist,graphicx,color,float,array,hyperref,amsmath,comment}
\usepackage[english]{babel}
\theoremstyle{theorem}
\newtheorem{theorem}{Theorem}[section]
\newtheorem{corollary}[theorem]{Corollary}
\newtheorem{lemma}[theorem]{Lemma}

\newtheorem{proposition}[theorem]{Proposition}
\theoremstyle{definition}
\newtheorem{definition}[theorem]{Definition}
\newtheorem{example}[theorem]{Example}
\newtheorem{remark}[theorem]{Remark}

\numberwithin{equation}{section}

\usepackage[normalem]{ulem}

\newcommand*{\dom}{{\rm dom}}

\newcommand*{\cl}{{\rm cl}}

\newcommand{\Reals}{\mathbb{R}}      
\newcommand{\Quot}{\mathbb{Q}}
\newcommand{\Erw}{\mathbb{E}}
\newcommand{\tn}{\textnormal}

\newcommand{\Ind}{\mathbf{1}}

\newcommand{\Prob}{\mathbb{P}}
\newcommand{\Lrho}{L^{\mathcal{R}}}
\newcommand{\Dualrho}{L^{\mathcal{R}*}}
\newcommand{\heart}{H^{\mathcal{R}}}

\newcommand{\Mrho}{M^{\mathcal{R}}}
\newcommand{\Rho}{\rho_{\mathcal R}}
\newcommand{\Norm}{\|\cdot\|}
\newcommand{\acc}{\mathcal A}
\newcommand{\N}{\mathbb N}
\newcommand{\Nat}{\mathbb{N}}
\newcommand{\eps}{\varepsilon}
\newcommand{\Bd}{\mathcal{L}^{\infty}}

\newcommand{\bounded}{\mathbf{ba}}
\newcommand{\countable}{\mathbf{ca}}
\newcommand{\cenv}{\mathbf{C}}
\newcommand{\Borel}{\mathbb{B}}
\newcommand{\Linfty}{L^{\infty}_{\mathbb{P}}}
\newcommand{\Pcal}{\mathcal P}
\newcommand{\Fcal}{\mathcal F}
\newcommand{\price}{\mathfrak p}
\newcommand{\Scal}{\mathcal S}
\newcommand{\Xcal}{\mathcal X}
\newcommand{\Ecal}{\mathcal E}
\newcommand{\tilrho}{\rho_{\tilde{\mathcal{R}}}}
\newcommand{\om}{\omega}
\DeclareMathOperator*{\esssup}{ess\,sup}

\title{\textsc{Model Spaces for Risk Measures}}
\date{\normalsize September 14, 2017}
\author{Felix-Benedikt Liebrich\thanks{E-mail: liebrich@math.lmu.de}\hspace{1.5cm}Gregor Svindland\thanks{E-Mail: svindla@math.lmu.de}\\
{\normalsize\textit{Department of Mathematics, LMU Munich, Germany}}}

\begin{document}
\maketitle

\abstract{We show how risk measures originally defined in a model free framework in terms of acceptance sets and reference assets imply a meaningful underlying probability structure. Hereafter we construct a maximal domain of definition of the risk measure respecting the underlying ambiguity profile. We particularly emphasise liquidity effects and discuss the correspondence between properties of the risk measure and the structure of this domain as well as subdifferentiability properties.\\ $\,$ 
}

\section{Introduction}

\noindent There is an ongoing debate on the {\em right} model space for financial risk measures, i.e.\ about what an ideal domain of definition for risk measures would be. Typically---as risk occurs in face of randomness---the risks which are to be measured are identified with real-valued random variables on some measurable space $(\Omega,\Fcal)$. The question which causes debate, however, is which space of random variables one should use as model space.\\
Since {\em risk} is often understood as Knightian \cite{Knight1921} uncertainty about the underlying probabilistic mechanism, many scholars argue that model spaces should be robust in the sense of not depending too heavily on some specific probabilistic model. We refer to this normative viewpoint as \textit{paradigm of minimal model dependence}. The literature usually suggests one of the following model spaces: 
\begin{itemize}
\item[(i)] ${\cal L}^0$ or $L^0_\Prob$, the spaces of all random variables or $\Prob$-almost sure ($\Prob$-a.s.) equivalence classes of random variables for some probability measure $\Prob$ on $(\Omega,\cal F)$, respectively, see \cite{Delbaen2000, 6};
\item[(ii)] $\cal L^{\infty}$ or $\Linfty$, the spaces of all bounded random variables or $\Prob$-a.s.\ equivalence classes of bounded random variables, respectively, see  \cite{Delbaen2000, 6, FS2002, 5, Kusuoka, Laeven1} and the references therein; 
\item[(iii)] $L^p_\Prob$, $p\in [1, \infty)$, the space of $\Prob$-a.s.\ equivalence classes of random variables with finite $p$-th moment, or more generally Orlicz hearts, see e.g.\ \cite{Laeven2, 9, Frittelli, Rockafellar2006}. 
\end{itemize}
The spaces in (i) and (ii) satisfy minimal model dependence in that ${\cal L}^{0}$ and $\cal L^{\infty}$ are completely model free, whereas $ L^0_{\Prob}$ and $\Linfty$ in fact only depend on the null sets of the probability measure $\Prob$. 
The problem with choosing ${\cal L}^0$ or $L^0_\Prob$, however, is that these spaces are in general too large to reasonably define aggregation based risk measures on them. The latter would require some kind of integral to be well-defined. Moreover, if $(\Omega, {\cal F})$ is not finite, ${\cal L}^0$ or $L^0_\Prob$  do not allow for a locally convex topology which make them unapt for optimisation. Important applications of risk measures, however, use them as objective functions or constraints in optimisation problems. Since $\cal L^{\infty}$ and $\Linfty$ are Banach spaces---so in particular locally convex spaces---and satisfy minimal model dependence, these model spaces have become most popular in the literature, and amongst them in particular  $\Linfty$ due to nicer analytic properties; see  \cite{Delbaen2000, 6, FS2002, 5, Kusuoka, Laeven1} and the references therein. In applications, however, unbounded models for risks are standard, like the log-normal distribution in Black-Scholes market models, etc. Assuming frictionless markets, there is no upper bound on the volumes and thus value of financial positions. Hence unbounded distributions appear quite naturally as limiting objects of bounded distributions, and in statistical modeling of random payoffs, where no upper bound can be assumed {\em a priori}. Also, risks with unbounded support and potentially heavy-tailed distributions are commonly employed in the insurance business. From this point of view model spaces should satisfy the \textit{paradigm of maximal domain} in that they should at least be sufficiently large to include these standard unbounded models, and the model spaces in (iii) have been proposed to resolve this issue. Problematic though is the strong dependence of $L^p_{\Prob}$, $p\in [1, \infty)$, (or in general   Orlicz hearts) on the probability measure $\Prob$ in that they are not invariant under equivalent changes of measure anymore. Consequently, maximal domain and minimal model dependence seem to be conflicting paradigms. \\
In the special case of law-invariant risk measures the measured risk is fully determined by the distribution of the risk under a probability measure $\Prob$  on $(\Omega,\cal F)$. Thus law-invariance already entails the existence of a {\em meaningful} reference probability model $\Prob$, and the risk measurement is fully depending on  $\Prob$. Hence, the ambiguity structure is such that it is no conceptual problem to define these risk measures on, for instance, $L^1_{\Prob}$ (see \cite{7}).\footnote{ In fact, law-invariant risk measures are completely unambiguous.} The latter observation shows that the paradigms of minimal model dependence and maximal domain may not be as conflicting as they seem, as long as the underlying probability structure is determined by the considered risk measure. Clearly, a model space like $\Linfty$ is sufficiently robust to carry any kind of risk measure. But given a specific risk measure---say defined on $L^\infty_\Prob$---and 
the corresponding ambiguity attitude reflected by it, a model space which respects this ambiguity attitude, which also carries the risk measure, and which is probably larger than $L^\infty_\Prob$, is also a reasonable model space for that particular risk measure---like in the (unambiguous) case of a law-invariant risk measure and the model space $L^1_\Prob$.\\  
Our starting point is an {\em a priori} completely model free setting on the model space ${\cal L}^\infty$ and a generalised notion of risk measurement adopted from Farkas, Koch-Medina and Munari in, e.g., \cite{FKM2015} and Munari in \cite{Munari}: all it requires is a notion of acceptability of losses (encoded by an acceptance set $\acc\subseteq\Bd$), a portfolio of liquidly traded securities allowed for hedging (represented by a subspace $\Scal\subseteq\Bd$), and a set of observable prices for these securities (a linear functional $\price$ on $\Scal$). Using such a risk measurement regime $\mathcal R=(\acc, \Scal,\price)$, we can define the risk $\Rho(X)$ to be the minimal price one has to pay in order to secure the loss $X\in\Bd$ with a portfolio in $\Scal$. This approach has the indisputable advantage of a clear operational interpretation. 
Section~\ref{sec:preliminaries} introduces this kind of risk measurement in a unifyingly general framework. In Section~\ref{sec:weakstrong}, we observe that under a standard approximation property of finite risk measures---continuity from above---they automatically imply a reference probability measure $\Prob$ which allows us to view the risk measure as defined on  $\Linfty$ without any loss of information. The observation that this often assumed property necessarily implies that the framework is dominated sheds new and critical light on the current discussion on model free and robust finance. Next, we demonstrate that under some further conditions, e.g., sensitivity and strict monotonicity, we can even find a \textit{strong} reference probability measure $\Prob^*\approx\Prob$ such that additionally 
\begin{center}$\forall X\in\Linfty:\quad \Rho(X)\geq c\Erw_{\Prob^*}[X]$\end{center} holds for a suitable constant $c>0$. These strong reference probability measures serve as a class of benchmark models in that risk estimation with $\Rho$ is uniformly more conservative than using the linear risk estimation rules $X\mapsto c\Erw_{\Prob^*}[X]$. \\
In Section~\ref{sec:minkowski:def}, we discuss how these considerations lead to a Banach space $\Lrho$ typically much larger than $\Linfty$, which has a multitude of desirable properties, such as  
\begin{itemize}
\item invariance under all strong and weak reference probability models;
\item a geometry completely determined by the risk measure $\Rho$;
\item robustness in that it carries an extension of the initial risk criterion $\Rho$, denoted by $\tilrho$, which preserves the functional form of $\Rho$, the dual representation, and thus convexity and lower semicontinuity. Moreover, this extension $\tilrho$ is a capital requirement in terms of unchanged hedging securities and pricing functionals, but with a notion of acceptability obtained by consistently extending constraints defining $\acc$ to $\Lrho$.
\end{itemize}
In the latter sense, $\Lrho$ can be seen as a natural maximal domain of definition of the initial risk criterion on which the ambiguity attitude is preserved.\\
We also consider the following monotone extensions of $\Rho$ to unbounded loss profiles in $\Lrho$ given by 
\begin{equation*}
\xi(X)=\lim_{m\to \infty} \lim_{n\to \infty} \Rho((-n)\vee X\wedge m)=\sup_{m\in \Nat} \inf_{n\in \Nat} \Rho((-n)\vee X\wedge m)
\end{equation*}
and
\begin{equation*}
\eta(X)=\lim_{n\to \infty} \lim_{m\to \infty} \Rho((-n)\vee X\wedge m)=\inf_{n\in \N} \sup_{m\in \N}  \Rho((-n)\vee X\wedge m)
\end{equation*}
which have been studied in for instance \cite{6,4}.
One would maybe expect that always $\tilrho=\xi=\eta$, but it turns out that  $\tilrho=\xi$ always holds, whereas $\tilrho\neq \eta$ is possible, see Example~\ref{ex:etaneqtilrho}. We characterise the often desirable regular situation when monotone approximation of risks in the following sense  
\begin{equation}\label{eq:regular}
 \tilrho(X)=\eta(X)=\lim_{n\to  \infty} \Rho((-n)\vee X\wedge n)
\end{equation}
is possible, see Theorem~\ref{thm:tilrho:eta}, and show that \eqref{eq:regular} holds if $\tilrho$ shows sufficient continuity in the tail of the risk $X$. For instance, any risk measure to which some kind of monotone or dominated convergence rule can be applied will satisfy \eqref{eq:regular}.
In Section~\ref{sec:structure}, we decompose $\Lrho$ into subsets with a clear interpretation in terms of liquidity risk and show how $\Lrho$ allows to view properties of the risk measure $\tilrho$ through a topological lens. Finally, in Section~\ref{sec:subgrad}, we address the issue of subdifferentiability of $\tilrho$ on $\Lrho$ based on a brief treatment of the dual of $\Lrho$ in Section~\ref{sec:dual}. Subgradients play an important role in risk optimisation and appear as pricing rules in optimal risk sharing schemes, see e.g.\ \cite{JST,4}. We shall see that the topology on $\Lrho$ being determined by $\Rho$ is fine enough to guarantee a rich class of points where $\tilrho$ is subdifferentiable, thereby further illustrating how suited the model space $\Lrho$ is to $\Rho$. Beside their mere existence, we also aim for reasonable conditions guaranteeing that subgradients correspond to measures on $(\Omega, \Fcal)$---which means ruling out singular elements that may exist in the 
dual space of $\Lrho$. The motivation for this is the same as in case of $\Linfty$ which in general also admits singular elements in its dual space. It is questionable whether such singular dual elements are reasonable as, for instance, pricing rules, because their effect lies mostly in the tails of the distribution, and the lack of countable additivity contradicts the paradigm of diminishing marginal risk. Also, measures show a by far better analytic behavior which may prove to be crucial when solving optimisation problems. Our findings suggest that singular elements do not really matter in a wide range of instances. In particular, we will also see that the local equality \eqref{eq:regular} characterised in Theorem~\ref{thm:tilrho:eta} is closely related to regular subgradients of $\tilrho$ and $\eta$. In Section~\ref{sec:examples} we collect illustrating examples. Some cumbersome proofs are outsourced to the appendices \ref{appendix:A} and \ref{appendix:B}.

\section{Some preliminaries}\label{sec:preliminaries}

\textit{Notation and terminology:} Given a set $M\neq\emptyset$ and a function $f: M\to[-\infty,\infty]$, we define the \textsc{domain} of $f$ to be the set $\dom(f)=\{m\in M\mid f(m)<\infty\}$. $f$ is called \textsc{proper} if it does not attain the value $-\infty$ and $\dom(f)\neq\emptyset$.  \\
For a subset $A$ of a topological space $(\mathcal X,\tau)$, we denote by $\cl_{\tau}(A)$ and $\tn{int}_{\tau}(A)$ the closure and interior of $A$, respectively, with respect to the topology $\tau$. If $(\mathcal X,\tau)$ is a topological vector space and $\tau$ is generated by a norm $\Norm$ on $\mathcal X$, we will replace the subscript $\tau$ by $\Norm$.\\
A triple $(\mathcal X,\tau,\preceq)$ is called \textsc{ordered topological vector space} if $(X,\tau)$ is a topological vector space and $\preceq$ is a partial vector space order compatible with the topology in that the positive cone of $\mathcal X$, denoted by $\mathcal X_+:=\{X\in \mathcal X\mid 0\preceq X\}$, is $\tau$-closed. We define $\mathcal X_{++}:=\mathcal X_+\backslash\{0\}$, and $\Xcal_-$ and $\Xcal_{--}$ analogously. If $(\Xcal,\tau,\preceq)$ is a Riesz space and $X, Y\in\Xcal$, we set $X\vee Y:=\sup\{X,Y\}$, $X\wedge Y:=\inf\{X,Y\}$, $X^+:=X\vee 0$, and $X^-:=(-X)\vee 0$.\footnote{ For details concerning ordered vector spaces, we refer to Chapters 5 and 7 of \cite{Ali}. Since risk measures will appear in this treatment on different domains of definition---in all cases spaces of random variables endowed with a pointwise or almost sure order and with varying topologies---we define them as functionals on ordered topological vector spaces. However, the reader may think of $\mathcal X$ as a space of 
random variables and of $\preceq$ as a pointwise or almost sure order on the latter.} 

\smallskip
\noindent In this section we define risk measurement regimes and risk measures, discuss some properties a risk measure may enjoy, and introduce the building blocks for a duality theory.
\begin{definition}\label{R1}Let $(\Xcal,\tau,\preceq)$ be an ordered topological vector space. An \textsc{acceptance set} is a non-empty proper and convex subset $\acc$ of $\Xcal$ which is monotone, i.e. $\acc-\Xcal_+\subseteq \acc$. A \textsc{security space} is a finite-dimensional linear subspace $\Scal\subsetneq\Xcal$ containing a non-null positive element $U\in \mathcal S\cap \Xcal_{++}$. We refer to the elements $Z\in\mathcal S$ as security portfolios, or simply securities. A \textsc{pricing functional} on $\Scal$ is a positive linear functional $\price: \mathcal S\to\Reals$ such that $\price(Z)>0$ for all $Z\in \mathcal S\cap \Xcal_{++}$.\\
A triple $\mathcal R:=(\acc, \mathcal S, \price)$ is a \textsc{risk measurement regime} if $\acc$ is an acceptance set, $\mathcal S$ is a security space and $\price$ is a pricing functional on $\Scal$ such that 
\begin{equation}\label{regime}\forall X\in \Xcal:\quad\sup\{\price(Z)\mid Z\in\Scal, X+Z\in\acc\}<\infty.\end{equation}
The \textsc{risk measure} associated to a risk measurement regime $\mathcal R$ is the functional 
\begin{equation}\label{defRM}\Rho: \Xcal\rightarrow(-\infty,\infty],\quad X\mapsto\inf\left\{\price(Z)\mid Z\in\mathcal S, X-Z\in\mathcal A\right\}. \end{equation}
\end{definition}
\noindent Our definition of risk measures is inspired by \cite{FKM2015, Munari}. Note that:
\begin{itemize}
\item[(a)] The elements $X\in \Xcal$ model \textit{losses}, not gains. Thus $\Rho(X)$ is the minimal amount which has to be invested today in some security portfolio $Z\in\Scal$ with payoff $-Z$ today in order to reduce the loss $X$ tomorrow to an acceptable level.
\item[(b)] We prescribe convexity of the acceptance set $\mathcal A$ which means that diversification is not penalised: if $X$ and $Y$ are acceptable so is the diversified $\lambda X+ (1-\lambda)Y$ for any $\lambda\in (0,1)$. 
\item[(c)] The notion of a risk measurement regime depends on the interplay of $\acc$, $\Scal$ and $\price$ by means of \eqref{regime}; this condition guarantees that $\Rho$ is a proper function. \cite[Propositions 1 and 2]{FKM2015} yield criteria for $\mathcal R$ to be a risk measurement regime in our sense. 
\end{itemize}
If $\mathcal S=\Reals\cdot U$ for some $U\in \Xcal_{++}$ and $\price(mU)=m$, $m\in\Reals$, the setting of \cite{FKM2013,FKM2014} with a single liquid eligible asset can be recovered from Definition \ref{R1}. If $\Xcal$ is a Riesz space with weak unit $\Ind$, $\mathcal S=\Reals\cdot\Ind$ and $\price(m\Ind)=m$, $m\in\Reals$, the definition covers {\em convex monetary risk measures} as comprehensively discussed in \cite{5}.\footnote{ In the following, we will refer to this particular case with the term \textit{monetary risk measures}.} The following is easily verified: 
\begin{lemma} Let $\mathcal R=(\acc,\Scal,\price)$ be a risk measurement regime on an ordered topological vector space $\Xcal$. Then $\Rho$ is convex, \textsc{monotone},  i.e.\ $X\preceq Y$ implies $\Rho(X)\leq\Rho(Y)$, and \textsc{$\Scal$-additive}, i.e.\ $\Rho(X+Z)=\Rho(X)+\price(Z)$ holds for all $X\in\Xcal$ and all $Z\in\Scal$.
\end{lemma}
\noindent In the same abstract setting we introduce further properties a risk measure can enjoy. 
\begin{definition}\label{properties1} 
Let $\mathcal R=(\acc,\Scal,\price)$ be a risk measurement regime on an ordered topological vector space $(\Xcal,\tau,\preceq)$ and let $\Rho$ be the associated risk measure.
\begin{itemize}
\item $\Rho$ is called \textsc{finite} if it only takes finite values, or equivalently $\acc+\Scal=\Xcal.~\footnote{ \cite[Propositions 1-3]{FKM2015} give further criteria to decide whether $\Rho$ is finite or not.}$
\item $\Rho$ is \textsc{normalised} if $\Rho(0)=0$, or equivalently $\sup_{Z\in\acc\cap\mathcal S}\price(Z)=0$.
\item $\Rho$ is  \textsc{coherent} if for any $X\in\Xcal$ and for any $t>0$ $\Rho(tX)=t\Rho(X)$ holds.
\item $\Rho$ is \textsc{sensitive} if it satisfies $\Rho(X)>\Rho(0)$ for all $X\in\Xcal_{++}$.
\item $\Rho$ is \textsc{lower semicontinuous} (l.s.c.) if every lower level set $\{X\in \Xcal\mid \Rho(X)\leq c\}$, $c\in\Reals$, is $\tau$-closed.
\item $\Rho$ is \textsc{continuous from above} if it is finite and for any $(X_n)_{n\in\Nat}\subseteq\Xcal$ with $X_n\downarrow X$ in order $\Rho(X)=\lim_n\Rho(X_n)$ holds.
\end{itemize}
\end{definition}

\begin{remark}\label{shapiro}
\begin{itemize}
\item[(i)]Normalisation implies that the negative cone $\Xcal_-$ (no losses) will be acceptable, which is economically sound. Every risk measure satisfying $\Rho(0)\in\Reals$ can be normalised by translating the acceptance set. Indeed, let $U\in\mathcal S\cap \Xcal_{++}$ and define $r:=\frac{\Rho(0)}{\price(U)}$ and $\check{\acc}:=\{X+rU\mid X\in\acc\}$. If $\mathcal R$ is a risk measurement regime, then so is $\check{\mathcal R}:=(\check\acc, \Scal,\price)$.
Moreover, \begin{center}$-\rho_{\check{\mathcal R}}(0)=\sup\{\price(Z)\mid Z\in\Scal, Z-rU\in\acc\}=\sup\{\price(W)+\price(rU)\mid W \in\Scal\cap\acc\}$.\end{center}
This implies that $-\rho_{\check{\mathcal R}}(0)=-\Rho(0)+\Rho(0)=0$ holds.
\item[(ii)] Recall that, in contrast to a large share of the literature on risk measures, random variables model losses, not gains, in our setting. Consequently, our notion of continuity from above is not the same as continuity from above in the sense of F\"ollmer and Schied (c.f. \cite[Lemma 4.21]{5}). The equivalent notion in the aforementioned monograph would be \textit{continuity from below} (c.f. \cite[Theorem 4.22]{5}), which together with lower semicontinuity of a risk measure implies the \textit{Lebesgue property}---see \cite{5}. 
\item[(iii)] Our notion of continuity from above means that approximating the risk of complex payoffs by the one of potentially easier but \textit{worse} financial instruments is meaningful as long as the payoffs range in a bounded regime.
\item[(iv)] Lower semicontinuity of $\Rho$ implies that $\{X\in \Xcal\mid\Rho(X)\leq 0\}=\cl_{\tau}(\acc+\ker(\price))$. In particular, it is implied by $\acc+\ker(\price)$ being closed (see \cite[Proposition 4]{FKM2015}) and invariant under translations of the acceptance set along $\Scal$. From an economic perspective this property is not too demanding: security spaces are always finite-dimensional in our setting, hence lower semicontinuity is, e.g.,  implied by the condition $\acc\cap\ker(\price)=\{0\}$ (cf. \cite[Proposition 5]{FKM2015}). The latter is sometimes referred to as \textit{absence of good deals of the first kind} (cf. \cite{Jaschke}).
\item[(v)] Note that in the case of $\Xcal$ being a Banach lattice with norm $\Norm$, every finite risk measure is norm-continuous and therefore also norm-l.s.c.\ This follows from \cite[Proposition 1]{Shapiro}: Suppose $\Xcal$ is a Banach lattice and $f:\Xcal\to(-\infty,\infty]$ is a proper convex and monotone function. Then $f$ is continuous on $\tn{int}\,\dom(f)$. We will make frequent use of this fact throughout the paper.
\end{itemize}
\end{remark}
\noindent For many questions a dual point of view on risk measures is crucial. In our case, its formulation requires the following concepts:
\begin{definition}\label{properties2}
Assume $\mathcal R=(\acc,\Scal, \price)$ is a risk measurement regime on an ordered topological vector space $(\Xcal,\tau,\preceq)$ with topological dual $\Xcal^*$. We define the \textsc{support function} of $\acc$ by
\begin{equation}\label{supportfunction}\sigma_{\acc}: \Xcal^*\to(-\infty,\infty],\quad \ell\mapsto\sup_{Y\in\acc}\ell(Y),\end{equation}
and $\mathcal B(\acc):=\dom(\sigma_{\acc})$. 
Moreover, the \textsc{extension set} will refer to the set of positive, continuous extensions of $\price$ to $\Xcal$, namely $\Ecal_{\price}:=\{\ell\in \Xcal^*_+\mid\,\ell|_{\Scal}=\price\}$.
\end{definition}
 
\section{Model spaces of bounded random variables, and weak and strong reference probability models}\label{sec:weakstrong}

\subsection{The model space  \texorpdfstring{$\mathbf{\Bd}$}{Bd} and weak reference probability measures}

Fix a measurable space $(\Omega,\mathcal{F})$ and let $\Bd:={\cal L}^\infty(\Omega, {\cal F})$ be the set of bounded measurable real-valued functions. We recall that $\Bd$ is a Banach lattice with norm $|X|_{\infty}:=\sup_{\omega\in\Omega}|X(\omega)|$ when equipped with the pointwise order $\leq$, so in particular an ordered topological vector space. On the level of Riesz spaces, $\Omega\ni\omega\mapsto 1$ is an order unit of $\Bd$.\footnote{ Recall that $e\in \Xcal_+$ is an order unit of a Riesz space $(\Xcal,\preceq)$ if $\{X\in \Xcal\mid\exists\lambda>0: |X|\preceq\lambda e\}=\Xcal$.} 
The dual space of $\Bd$ may be identified with $\bounded$, the space of all finitely additive set functions $\mu:{\mathcal F}\to \Reals$. As usual $\countable$ denotes the countably additive set functions in $\bounded$, and $\countable_+$ is the set of finite measures. In the following, the notation will not distinguish between $m\in\Reals$ and the function $\Omega\ni\omega\mapsto m$.\\
In this section we study risk measures on $\Bd$. First of all, note that in the $\langle \Bd,\bounded\rangle$-duality monotonicity of $\acc$ implies that $\mathcal B(\acc)\subseteq \bounded_+$ has to hold, and an application of the Hahn-Banach Separation Theorem shows \begin{equation}\label{acceptance}\cl_{|\cdot|_{\infty}}(\acc)=\{X\in\Bd\mid\forall\mu\in\mathcal B(\acc):~\int X\,d\mu\leq\sigma_{\acc}(\mu)\}.\end{equation}
We will mostly assume finiteness of $\Rho$, which is justified by the domain of definition ${\cal L}^\infty$---that is {\em bounded} losses which typically should be hedgeable at potentially large, but finite cost. $\Rho$ is for instance finite whenever the security space $\Scal$ contains some $U\in\Bd_{++}$ being uniformly bounded away from $0$, i.e. $U \geq \delta$ for some constant $\delta>0$. In \cite{FKM2013,FKM2014}, such securities are called \textsc{non-defaultable}. We will show that if the acceptance set is ``nice enough'', then any finite risk measure arising from it in an \textit{a priori} model-free framework like $\Bd$ indeed implies a probabilistic model, a so-called weak reference model; see Theorem~\ref{R2}. As a first step towards this result, we show now that continuity from above mainly depends on the geometry of the acceptance set $\acc$.
 To this end, let us recall the notion of the \textsc{dual conjugate} of $\Rho$ being defined as 
\begin{equation}\label{eq:dualrho}\Rho^*:\bounded\to (-\infty,\infty],\quad\mu\mapsto\sup_{X\in\Bd}\int X\, d\mu-\Rho(X).\end{equation}
\begin{proposition}\label{ancillary}
Assume $\mathcal R=(\acc,\Scal,\price)$ is a risk measurement regime such that $\Rho(0)\in\Reals$.
\begin{enumerate}
\item[(i)] If $\Rho$ is l.s.c., $\mathcal B(\acc)\cap\mathcal E_{\price}$ is non-empty and for all $\mu\in\bounded$ it holds that
\begin{equation}\label{support}\Rho^*(\mu)=\begin{cases}\sigma_{\acc}(\mu)\quad\textit{if }\mu\in\mathcal B(\acc)\cap\Ecal_{\price},\\
\infty\quad\textit{otherwise}.\end{cases}\end{equation}
For all $X\in \Bd$ we have \begin{equation}\label{eq:dualrep1} \Rho(X)=\sup_{\mu\in\dom(\Rho^*)}\int X\, d\mu-\Rho^*(\mu). \end{equation}
Moreover, if $\Rho$ is coherent, then \begin{equation}\label{support:coh}\Rho^*(\mu)=\begin{cases}0\quad\textit{if }\mu\in\mathcal B(\acc)\cap\Ecal_{\price},\\
\infty\quad\textit{otherwise}, \end{cases}\end{equation} and \begin{equation}\label{eq:dualrep:coh} \Rho(X)=\sup_{\mu\in\dom(\Rho^*)}\int X\, d\mu, \quad X\in \Bd.  \end{equation}
\item[(ii)] Assume $\mathcal R=(\acc,\Scal,\price)$ is a risk measurement regime such that $\Rho$ is finite, then for every $c\in\Reals$ the lower level set $E_c:=\{\mu\in\bounded\mid\Rho^*(\mu)\leq c\}$ of $\Rho^*$ is $\sigma(\bounded,\Bd)$-compact. 
\item[(iii)]Suppose the risk measure $\Rho$ associated to the risk measurement regime $\mathcal R=(\acc,\Scal,\price)$ is finite. Then $\Rho$ is continuous from above if and only if every lower level set $E_c$, $c\in\Reals$, of $\Rho^*$ is $\sigma(\countable,\Bd)$-compact. Hence, if $\mathcal B(\acc)\subseteq\countable$, then $\Rho$ is continuous from above. 
\item[(iv)] In the situation of (iii), if $\Scal$ is constrained to be one-dimensional, then $\Rho$ is continuous from above if and only if $\mathcal B(\acc)\subseteq\countable$. 
\end{enumerate}
\end{proposition}
\begin{corollary}\label{cor:ancillary}Assume $\mathcal R=(\acc,\Scal,\price)$ is a risk measurement regime such that $\Rho$ is finite. Then $\Rho$ is continuous from above if and only if $\mathcal B(\acc)\cap \ker(\price)^\perp\subseteq\countable$. Here, 
$$\ker(\price)^\perp:=\left\{\mu\in\bounded\Big|\,\forall Z\in\ker(\price):\,\int Z\,d\mu=0\right\}$$
denotes the annihilator of $\ker(\price)$. 
\end{corollary}

\noindent For the special case of a monetary risk measure, parts of Proposition~\ref{ancillary} are well-known, see e.g.\ \cite[Theorem 4.22 and Corollary 4.35]{5}. However, to our knowledge, so far there is no proof of Proposition~\ref{ancillary} and Corollary~\ref{cor:ancillary} in this general form in the literature. As the proofs of these results are quite technical and thus lengthy we provide them in Appendix~\ref{appendix:A}. Note that the representation \eqref{eq:dualrep1} is in terms of pricing rules consistent with $(\Scal,\price)$ in that $\mu \in \dom(\Rho^*)$ only if $\mu|_{\Scal}=\price$. If $S=\Reals$ and $\price=id_{\Reals}$, these pricing rules can be identified with probability measures. Finally, we remark that continuity from above is indeed mainly a property of the acceptance set as Proposition~\ref{ancillary}(iii) and (iv) and Corollary~\ref{cor:ancillary} show: If $\acc$ is regular in the sense that $\mathcal B(\acc)\subseteq \countable$, then every finite risk measure is continuous from above. In particular, taking a single hedging asset or multiple ones will have no effect on continuity from above provided $\acc$ is properly chosen. Non-regularity of the acceptance set in that $\mathcal B(\acc)\backslash \countable\neq\emptyset$, however, is equivalent to the fact that no finite risk measure with a single security is continuous from above; higher-dimensional security spaces may smooth out the irregularity of $\acc$, as illustrated in Example \ref{ex:counterex}. \\
The following theorem is the already advertised main result of this section. As facilitating notation, for non-empty sets of set functions $M, M'\subseteq\bounded$, we write $M\ll M'$ if and only if $\nu(A)=0$ for all $\nu\in M'$ implies $\mu(A)=0$ for all $\mu\in M$, $A\in\cal F$. We set $M\approx M'$ to mean that both $M\ll M'$ and $M'\ll M$. Instead of $\{\mu\}\ll\{\nu\}$ or $\{\mu\}\approx\{\nu\}$, we shall write $\mu\ll \nu$ and $\mu\approx \nu$. Finally, we define $\bounded_\nu:=\{\mu\in\bounded\mid\mu\ll\nu\}$, and $\countable_\nu$ analogously.
\begin{theorem}\label{R2} Let $\mathcal R=(\acc,\Scal,\price)$ be a risk measurement regime such that $\Rho$ is finite and continuous from above. 
\begin{enumerate}
\item[(i)] There exists a \textsc{weak reference probability measure} $\Prob$, that is a probability measure $\Prob$ on $(\Omega, {\cal F})$ such that $\Rho^*(c\Prob)<\infty$ for a suitable $c>0$ and 
\begin{equation}\label{100}\Prob \approx \dom(\Rho^*).\end{equation}
\item[(ii)] For $\Prob$ as in (i) we have that $\dom(\Rho^*)\subseteq (\countable_{\Prob})_+$. 
\item[(iii)] If $\Rho$ is normalised, then $E_0=\{\mu \in \countable \mid \Rho^* (\mu)=0\}\neq \emptyset$.
\end{enumerate}
\end{theorem}
\begin{proof}
For (i), recall from Proposition~\ref{ancillary} that the assumption on $\Rho$ implies that any lower level set $E_c:=\{\mu\in\countable_+\mid\Rho^*(\mu)\leq c\}$, $c\in\Reals$, is $\sigma(\countable,\Bd)$-compact. Together with convexity, this implies \textit{countable convexity}, i.e. 
\begin{equation}\label{eq:countconv}(\lambda_k)_{k\in\Nat}\subseteq [0,1],~\sum_{k=1}^\infty\lambda_k=1,~(\mu_k)_{k\in\Nat}\subseteq E_c~\Longrightarrow~\sum_{k=1}^\infty\lambda_k\mu_k\in E_c.\end{equation}
By \cite[Theorem 4.7.25,(iv) $\Rightarrow$ (i)]{13}, $E_n$, $n\in\Nat$, also has compact closure in the weak topology $\sigma(\countable,\countable^*)$. As $E_n$ is already closed in the weaker topology $\sigma(\countable,\Bd)$, $E_n$ has to be weakly compact. The proof of \cite[Theorem 4.7.25, (i) $\Rightarrow$ (ii)]{13} shows the existence of a sequence $(\mu^n_l)_{l\in\Nat}\subseteq E_n$ such that $E_n\approx\{\mu_l^n\mid l\in\Nat\}$. We set $\nu_n:=\sum_{l=1}^\infty2^{-l}\mu_l^n$, which lies in $E_n$ by \eqref{eq:countconv}, and satisfies $\nu_n\approx E_n$. By \eqref{eq:dualrep1}, the sequence $(\nu_n)_{n\in\Nat}$ satisfies $\nu_n(\Omega)\leq \nu_n(\Omega)-\Rho^*(\nu_n)+n\leq\Rho(1)+n$. Define 
\[\nu:=\sum_{n\in\Nat}2^{-n}\nu_n,\quad c_N:=\sum_{n=1}^N2^{-n},\quad \zeta_N:=c_N^{-1}\sum_{n=1}^N2^{-n}\nu_n,\quad N\in\Nat.\]
$\nu\in\countable_+$ follows from the estimate $\nu(\Omega)\leq \sum_{n=1}^\infty2^{-n}(\Rho(1)+n)=\Rho(1)+2.$
Every non-trivial scalar multiple of $\nu$ satisfies \eqref{100}, and moreover, $\nu=\lim_N\zeta_N$ with respect to $\sigma(\countable, \mathcal L^{\infty})$. Lower semicontinuity and convexity of $\Rho^*$ and $\Rho^*(\nu_n)\leq n$ imply
\[\Rho^*(\nu)\leq \liminf_{N\to\infty}\Rho^*(\zeta_N)\leq\lim_{N\to\infty}c_N^{-1}\sum_{n=1}^N2^{-n}n=2.\]
Choosing $c:=\nu(\Omega)$, the probability measure $\Prob:=\frac 1 c\nu$ is a weak reference probability model. \\
(ii) is an immediate consequence of \eqref{100}. In order to prove (iii) note that normalisation implies $0=\Rho(0)=-\inf\{\Rho^*(\mu)\mid \mu\in\dom(\Rho^*)\}$. Hence, $\Rho^*\geq 0$ and the family of subsets $(E_k)_{k\in (0,1]}$ of the compact set $E_1$ has the finite intersection property. Therefore $E_0=\bigcap_{k\in (0,1]}E_k\neq \emptyset$.
\end{proof}

\begin{remark}Continuity from above is sufficient but not necessary for the existence of weak reference probability models. However, without continuity from above anything can happen. For example, let $(\Omega,\Fcal)$ be the open unit interval $(0,1)$ endowed with its Borel sets $\Borel((0,1))$, and let $\Prob$ be the Lebesgue measure on $(0,1)$. Consider  
\[\esssup(X):=\sup\{m\in\Reals\mid \Prob(X\leq m)=1\},\]
and set 
$$\acc_1:=\{X\in\Bd\mid \esssup(X)\leq 0\},\quad\acc_2:=\{X\in\Bd\mid\sup_{\om\in\Omega}X(\om)\leq 0\}.$$ 
The triples $\mathcal R_i=(\acc_i,\Reals,id_{\Reals})$, $i=1,2$, are risk measurement regimes. In the first case, $\dom(\rho_{\mathcal R_1}^*)=(\bounded_\Prob)_+$, and in the second $\dom(\rho_{\mathcal R_2}^*)=\bounded_+$. Thus neither $\rho_{\mathcal R_1}$ nor $\rho_{\mathcal R_2}$ is continuous from above. $\Prob$, however, is a weak reference probability model for $\rho_{\mathcal R_1}$, whereas in the case of $\rho_{\mathcal R_2}$  there is no weak reference probability model as $\Omega$ is uncountable.
\end{remark}

\noindent Whenever a probability measure $\Prob$ satisfies \eqref{100} and $X,Y\in \Bd$ are equal $\Prob$-almost surely ($\Prob$-a.s.), \eqref{eq:dualrep1} shows that $\Rho(X)=\Rho(Y)$. Hence, we may view $\Rho$ as a function on the space of equivalence classes $\Linfty:=L^\infty(\Omega, {\mathcal F}, \Prob)$ with the corresponding properties. We recall that the least upper bound $$\|X\|_{\infty}:=\inf\{m\in \Reals\mid\,\Prob(|X|\leq m)=1\}, \quad X\in \Linfty, $$ is a norm on $\Linfty$, making it into a Banach lattice together with the $\Prob$-almost sure order, and the equivalence class generated by $\Omega\ni\omega\mapsto 1$ is an order unit of $\Linfty$. Its dual may be identified with $\bounded_\Prob$. Let $\iota: \Bd\rightarrow\Linfty$ be the canonical embedding, then it is straightforward to prove the following result. 

\begin{corollary}\label{R5}
In the situation of Theorem~\ref{R2} define $\rho: \Linfty\rightarrow\Reals$ by $\rho(\tilde X)=\Rho(X)$, where $X\in \Bd$ satisfies $\iota(X)=\tilde X$. Then $\rho$ is well-defined and agrees with the risk measure $\rho_{(\iota(\acc),\iota(\Scal),\bar\price)}$ on $\Linfty$, where $\bar\price(\tilde Z)=\price(Z)$ whenever $\tilde Z=\iota(Z)$. It is norm-continuous and continuous from above. The dual function \begin{equation}\label{eq:dual}\rho^*(\mu):=\sup_{ \tilde X\in\Linfty}\int X\, d\mu -\rho(\tilde X), \quad \mu\in \bounded_\Prob, \end{equation}  where $X$ denotes an arbitrary representative of $\tilde X$, agrees with $\Rho^*|_{\bounded_\Prob}$. Also $$\rho(\tilde X)=\sup_{\mu\in \dom(\Rho^*)}\int  X\, d\mu - \Rho^*(\mu),\quad \tilde X\in \Linfty,$$where $X$ and $\tilde X$ are related as before.
\end{corollary}

\subsection{The model space \texorpdfstring{$\mathbf{\Linfty}$}{Linfty} and strong reference probability measures}\label{sec:strong:ref}

Supported by our results in Proposition~\ref{R2} and Corollary~\ref{R5} we will from now on consider the model space $\Linfty$, acceptance sets $\acc\subseteq\Linfty$, security spaces $\mathcal S\subseteq\Linfty$, pricing functionals $\price: \mathcal S\to\Reals$ and resulting finite risk measures $\Rho$ directly defined on $\Linfty$, where $\Prob$ is a weak reference probability model for $\Rho$.  Moreover, in the following we will stick to the usual convention of identifying an equivalence class of random variables in $\Linfty$ with an arbitrary representative of that class.\\
By similar reasoning as in Proposition \ref{ancillary} and Theorem~\ref{R2} we have the following result. 

\begin{lemma}\label{lem:strongref}Let ${\cal R}=(\acc,\Scal,\price)$ be a risk measurement regime on $\Linfty$ such that $\Rho$ is finite and normalised. Define \begin{equation}\label{eq:dualrep3}\Rho^*(\mu):=\sup_{X\in\Linfty}\int X\, d\mu - \Rho(X), \quad \mu\in \bounded_\Prob. \end{equation} Then $\Rho$ is continuous from above if and only if all lower level sets $E_c:=\{\mu\in\bounded_\Prob\mid\Rho^*(\mu)\leq c\}$, $c\in\Reals$, of $\Rho^*$ are $\sigma(\countable_\Prob, \Linfty)$-compact, which is in particular implied by
\[\mathcal B(\acc):=\left\{\mu\in\bounded_\Prob\Big| \sup_{Y\in\acc}\int Y\,d\mu<\infty\right\}\subseteq \countable_{\Prob}.\]
In that case 
\begin{equation}\label{eq:dualrep2}\Rho(X)=\sup_{\mu\in (\countable_{\Prob})_+}\int  X\, d\mu - \Rho^*(\mu),\quad  X\in \Linfty.\end{equation} In particular, $\dom(\Rho^*)=\mathcal B(\acc)\cap\{\nu\in(\countable_\Prob)_+\mid\,\forall Z\in\Scal: \int Z\,d\nu=\price(Z)\}$, and $E_0\neq \emptyset$. If $\Rho$ is positively homogeneous, then the analogues of \eqref{support:coh} and \eqref{eq:dualrep:coh} hold as well.
\end{lemma}
\noindent We devote the remainder of this section to the question whether there is a \textsc{strong reference model}, i.e. whether there an element in 
$$\Pcal:=\{\mu\in E_0\mid \mu\approx \Prob\}.$$ This notion is well-known in the case of law-invariant monetary risk measures, and the result can be generalised in our setting: 
\begin{proposition}Let  ${\cal R}=(\acc,\Scal,\price)$ be a risk measurement regime on $\Linfty$ such that $\Rho$ is normalised, and assume
\begin{enumerate}
\item[(i)] the underlying probability space $(\Omega,\Fcal,\Prob)$ is atomless;
\item[(ii)] $\acc$ is the acceptance set $\{X\in\Linfty\mid \mathfrak r(X)\leq 0\}$ of a normalised, $\Prob$-law-invariant monetary risk measure $\mathfrak r$ which is continuous from above;
\item[(iii)] $\price=c\Erw_\Prob[\cdot]$ for a suitable constant $c>0$. 
\end{enumerate}
Then $\Prob\in\mathcal P$. 
\end{proposition}
\begin{proof}
By \cite[Proposition 1.1]{continuity} and \cite[Corollary 4.65]{5}, $\mathfrak r$ is dilatation monotone: for every sub-$\sigma$-algebra $\mathcal G\subseteq\Fcal$ and every $X\in\Linfty$, the estimate $\mathfrak r(\Erw_{\Prob}[X|\mathcal G])\leq\mathfrak r(X)$ holds. Thus, for every $X\in\Linfty$, $\Erw_{\Prob}[X]=\mathfrak r\left(\Erw_{\Prob}[X|\{\emptyset,\Omega\}]\right)\leq \mathfrak r(X)$. We conclude $\sigma_\acc(\Prob)=\sup_{Y\in\Linfty:~\mathfrak r(Y)\leq 0}\Erw_\Prob[Y]=0$. 
\end{proof}
\noindent Clearly, sensitivity (c.f. Definition \ref{properties1}) is necessary to have $\Pcal\neq \emptyset$, but apart from the coherent case it is not sufficient. As an example consider two probability measures $\Quot\ll\Prob$ such that $\Quot\not \approx \Prob$. Define
\[\Prob_{\beta}:=\beta\Quot+(1-\beta)\Prob,\quad\rho:\Linfty\ni X\mapsto\sup_{\beta\in[0,1]}\Erw_{\Prob_{\beta}}[X]-(1-\beta)^2.\]
For $\mathcal R=(\{X\mid \rho(X)\leq 0\}, \Reals, id_{\Reals})$, we have that $\rho=\Rho$ is a sensitive risk measure with $E_0=\{\Quot\}$ and $\Pcal=\emptyset$. \\
In the following we will use the notation $\Fcal_+:=\{A\in\Fcal\mid\Prob(A)>0\}$. 
\begin{lemma}\label{lem:cohpcal}
Let $\mathcal R=(\acc,\Scal,\price)$ be a risk measurement regime such that $\Rho$ is finite, continuous from above, and coherent. Then $\Pcal\neq \emptyset$ if and only if $\Rho$ is sensitive.
\end{lemma}

\begin{proof} We only prove sufficiency.
If $\Rho$ is coherent, then $\dom(\Rho^*)=E_0$; see \eqref{support:coh}. Moreover, continuity from above implies $E_0\subseteq (\countable_\Prob)_+$. 
As $\Rho$ is sensitive, we have that $0<\Rho(\Ind_A)=\sup_{\mu\in E_0}\mu(A)$ for all $A\in \Fcal_+$. Consequently, there is $\mu_A\in E_0$ such that $\mu_A(A)>0$. In other words, $E_0\approx \Prob$. The Halmos-Savage Theorem \cite[Theorem~1.61]{5} shows that there is a countable family $(\mu_n)_{n\in \N}\subseteq E_0$ such that $\{\mu_n\mid n\in \N\}\approx \Prob$. \eqref{eq:countconv} ensures that also $\nu:=\sum_{n\in \Nat}\frac{1}{2^n}\mu_n\in E_0$, i.e.\ $\Pcal\neq \emptyset$.
\end{proof}

\noindent The following theorems state sufficient conditions under which $\cal P\neq\emptyset$ without requiring coherence of $\Rho$. First, we characterise the strong condition $E_0=\Pcal$ with the ability of $\Rho$ to identify arbitrage.

\begin{theorem}\label{R10} Let $\mathcal R=(\acc,\Scal,\price)$ be a risk measurement regime, and suppose that $\Rho$ is  finite, normalised, continuous from above, and sensitive. 
The following are equivalent:
\begin{enumerate}
\item[(i)] $E_0=\Pcal$;
\item[(ii)] For all $A\in\mathcal F_+$ we have $\Rho(-k\Ind_A)<0$ for $k>0$ sufficiently large;
\item[(iii)] For all $X\in(\Linfty)_{++}$ we have $\Rho(-X)<0$.
\end{enumerate}
Moreover, $E_0=\Pcal$ if $\Rho$ is strictly monotone, i.e. $\Rho(X)<\Rho(Y)$ whenever $Y-X\in (\Linfty)_{++}$.
\end{theorem}

\begin{proof}
(iii) trivially implies (ii). Now assume (i) does not hold, i.e. there is some $\mu\in E_0\backslash\Pcal$, hence $\mu(A)=0$ for some $A\in\mathcal F_+$. From $0=\Rho(0)\geq\Rho(-k\Ind_A)\geq-k\mu(A)=0$ we infer $\Rho(-k\Ind_A)=0$ for all $k>0$, 
contradicting (ii). This shows that (ii) implies (i). In order to show that (iii) is implied by (i), assume we can find a $X\neq 0$ in the negative cone with $\Rho(X)=0$. As the level sets of $\Rho^*$ are $\sigma(\countable_\Prob,\Linfty)$-compact, we can find a $\mu\in\dom(\Rho^*)$ such that $0=\Rho(X)=\int X\,d\mu-\Rho^*(\mu)$. This implies $\Rho^*(\mu)=0=\int X\,d\mu$, a \textsc{contradiction} to $E_0=\mathcal P$. Finally, strict monotonicity clearly implies (iii) by normalisation.
\end{proof}

\noindent The next aim is a characterisation of $\Pcal\neq \emptyset$ in terms of the components of the risk measurement regime $\mathcal R=(\acc,\Scal,\price)$. 

\begin{theorem}\label{R1111}
Suppose that $\Rho$ is finite, continuous from above, normalised, and sensitive. Let $\cenv\subseteq\Linfty$ be the smallest weakly* closed convex cone containing $\acc+\ker(\price)$.
\begin{enumerate}
\item[(i)] $\cal P\neq \emptyset$ if and only if $~\cenv\cap (\Linfty)_{++}=\emptyset$.
\item[(ii)] $\Pcal\neq \emptyset$ if 
 $\acc+\ker(\price)$ satisfies the \textsc{rule of equal speed of convergence}: Let $(X_n)_{n\in\Nat}\subseteq\acc$ and $(Z_n)_{n\in\Nat}\subseteq\ker(\price)$ be sequences such that $\|X_n+Z_n\|_{\infty}\leq 1$ for all $n$. Suppose $(t_n)_{n\in\Nat}$ is such that $t_n\uparrow \infty$. If the rescaled vectors 
\[V_n:=t_n(X_n+Z_n)\]
satisfy $V_n^-\rightarrow 0$ in probability, then for all sets $B\in\Fcal_+$, it holds that 
\[\limsup_{n\rightarrow\infty}\Prob(B\cap\{V_n^+\geq \eps\})<\Prob(B).\]
\item[(iii)] $\cal P=\emptyset$ if there are sequences $(X_n)_{n\in\Nat}, (Z_n)_{n\in\Nat}$ and $(t_n)_{n\in\Nat}$ such that 
\[\sup_{n\in\Nat}\|t_n(X_n+Z_n)\|_{\infty}<\infty\]
violating the rule of equal speed of convergence 
\end{enumerate}
\end{theorem}
\noindent A proof is given in Appendix \ref{appendix:B}.

\section{The Minkowski domain of a risk measure}\label{sec:Minkowski}

\subsection{Construction of the Minkowski domain and extension results}\label{sec:minkowski:def}

Throughout Section~\ref{sec:Minkowski} fix an acceptance set $\acc\subseteq\Linfty$, a security space $\Scal\subseteq \Linfty$, and let $\price:\Scal\to\Reals$ be a pricing functional such that $\Rho:\Linfty\to \Reals$ is a normalised finite risk measure which is continuous from above. Based on the results in Section~\ref{sec:weakstrong}, we  assume that $\Prob$ is a weak reference probability model, i.e. $\gamma\Prob\in\dom(\Rho^*)$ for a suitable constant $\gamma>0$. The aim of this section is to lift $\Rho$ to a domain of definition denoted by $\Lrho$ whose structure is completely characterised by $\Rho$ and thus consistent with the initial risk measurement regime, although it is in general strictly bigger than $\Linfty.~$The typical argument for restricting risk measures to bounded random variables---namely, that this space is robust and thus not conflicting with the ambiguity expressed by $\Rho$---is not valid in this case, since $\Lrho$ will completely reflect the ambiguity as perceived
  by $\Rho$.
To this end, we remark that \begin{equation}\label{eq:g1} \rho(|X|):=\sup_{\mu\in \dom(\Rho^*)}\int |X|\,d\mu-\Rho^*(\mu),\end{equation} where $\Rho^\ast$ is given in \eqref{eq:dualrep3}, is well-defined for all $X\in L^0_\Prob:=L^0(\Omega, {\cal F}, \Prob)$, possibly taking the value $\infty$. In this sense the objects appearing in the following definition are well-defined.

\begin{definition}For $c>0$ and $X\in L^0_\Prob$ let
\begin{center}$\|X\|_{c,\mathcal R}:=\inf\left\{\lambda>0\Big|\,\rho\left(\frac{|X|}\lambda\right)\leq c\right\}\quad (\inf\emptyset:=\infty),$\end{center}
and $\|X\|_{\mathcal R}:=\|X\|_{1,\mathcal R}$. The \textsc{Minkowski domain} for $\Rho$ is the set
\[\Lrho:=\{X\in  L^0_\Prob\mid\|X\|_{\mathcal R}<\infty\}.\]
\end{definition}

\noindent
Note that we may interpret $\Norm_{\mathcal R}$ as a Minkowski functional given the level set
\begin{center}$\rho(|\cdot|)^{-1}(-\infty,1]$,\end{center} and its domain $\Lrho$ is thus called the \textit{Minkowski domain}.

\begin{proposition}\label{R12} 
\begin{enumerate}
\item[(i)] For all $c>0$ there exist constants $A_c,B_c>0$ such that $$A_c \Norm_{c,\mathcal R}\leq \Norm_{\mathcal R}\leq B_c \Norm_{c,\mathcal R}.$$ In particular $\Lrho=\{X\in L^0_\Prob\mid \|X\|_{c,\mathcal R}<\infty\}$ for all $c>0$, and $(\Norm_{c,\mathcal R})_{c>0}$  is a family of equivalent norms on $\Lrho$.\\
Moreover, $\|X\|_\infty\geq B_{\Rho(1)}^{-1}\|X\|_{\mathcal R}$, $X\in L^\infty_\Prob$, and thus $\Linfty\subseteq \Lrho$.
\item[(ii)]$X\in\Lrho$ if and only if $\int |X|\,d\mu<\infty$ for all $\mu\in\dom(\Rho^*)$.
\item[(iii)] $(\Lrho,\Norm_{\mathcal R})$ is a Banach lattice.
\item[(iv)] $|X|\leq |Y|$ implies $\|X\|_{c,\mathcal R}\leq\|Y\|_{c,\mathcal R}$ and thus $\Lrho$ is solid. In particular, $\Lrho$ is invariant under rearrangements of profits and losses, i.e.\ if $\varphi\in\Linfty$ attaining values in $[-1,1]$, then $\varphi\cdot X\in\Lrho$ with $\|\varphi X\|_{c,\mathcal R}\leq\|X\|_{c,\mathcal R}.$
\end{enumerate}
\end{proposition}
\begin{proof}
First we set $\Lambda_c(X):=\{\lambda>0\mid\rho(\lambda^{-1}|X|)\leq c\}$, i.e. $\|X\|_{c,\mathcal R}=\inf\Lambda_c(X).$

\smallskip\noindent
(i): Suppose that $c\in (0,1)$ and let $X\in L^0_\Prob$. Note that $\|X\|_{\mathcal R}=\infty$ if and only if $\Lambda_1(X)=\emptyset$, which implies $\Lambda_c(X)=\emptyset$ or equivalently $\|X\|_{c,\mathcal R}=\infty$. Now assume $\|X\|_{\mathcal R}<\infty$, and pick $\lambda\in\Lambda_1(X)$. As $\Rho^*\geq 0$, we have $$\rho(c|X|/\lambda)=\sup_{\mu\in \dom(\Rho^*)}\int \frac{c}{\lambda}|X|\,d\mu-\Rho^*(\mu)\leq c\rho(|X|/\lambda)\leq c,$$ which implies $\|X\|_{\mathcal R}\geq c \|X\|_{c,\mathcal R}$. Trivially, $\Lambda_c(X)\subseteq\Lambda_1(X)$ and therefore $\|X\|_{c,\mathcal R}\geq \|X\|_{\mathcal R}$. Hence, we may choose $A_c=c$ and $B_c=1$. The case $c>1$ is treated similarly.\\
Monotonicity implies that $\rho(|X|/\|X\|_\infty)\leq \Rho(1)$ for all $X\in \Linfty$, which yields $\|X\|_\infty\geq \|X\|_{\Rho(1),\mathcal R}\geq B_{\Rho(1)}^{-1}\|X\|_{\mathcal R}$ and $\Linfty\subseteq \Lrho$.\\
$\Norm_{c,\mathcal R}$ is indeed a norm: The verification of the triangle inequality and homogeneity are straightforward. For the definiteness of $\Norm_{c,\mathcal R}$, let $\mu\in\dom(\Rho^*)$ be arbitrary. As for all $\lambda\in\Lambda_c(X)$ we obtain the estimate $\|\lambda^{-1}X\|_{L^1_\mu}-\Rho^*(\mu)\leq\rho(\lambda^{-1}|X|)\leq c$, we can infer
\begin{equation}\label{16}\frac 1{c+\Rho^*(\mu)}\|X\|_{L^1_\mu}\leq\|X\|_{c,\mathcal R}.\end{equation}
Choosing $\mu\in\dom(\Rho^*)$ such that  $\mu=\gamma\Prob$ yields definiteness of $\Norm_{c,\mathcal R}$.

\smallskip
\noindent(ii): It follows from \eqref{16} that for all $X\in\Lrho$ and all $\mu\in\dom(\Rho^*)$ the integrability condition $\int |X|\,d\mu<\infty$ holds. For the converse implication, let $X\in L^0_\Prob\backslash \Lrho$ be arbitrary, the latter being equivalent to $\rho(t|X|)=\infty$ for all $t>0$. As before, we set $E_c:=\{\mu\in\countable_\Prob\mid \Rho^*(\mu)\leq c\}$, $c\in\Reals$, and will show that there is a $\nu\in E_1$ such that $\int |X|\,d\nu=\infty$. First assume that $\sup_{\mu\in E_1}\int|X|\,d\mu=\infty$. Choose a sequence $(\mu_n)_{n\in\Nat}\subseteq E_1$ such that $\int |X|\,d\mu_n\geq 2^{2n}$, $n\in\Nat$, and set $\nu=\sum_{n=1}^\infty2^{-n}\mu_n$, which is itself an element of $E_1$ by \eqref{eq:countconv}. Moreover, 
$$\int |X|\,d\nu=\sum_{n=1}^\infty2^{-n}\int |X|\,d\mu_n\geq\sum_{n=1}^\infty2^n=\infty.$$
Hence, $X$ is not $\nu$-integrable. In a second step, we show that the case $\sup_{\mu\in E_1}\int |X|\,d\mu<\infty$ cannot occur. Assume for contradiction that $\sup_{\mu\in E_1}\int |X|\,d\mu<\infty$. If there were a constant $\kappa>0$ such that for all $\mu\in\dom(\Rho^*)\backslash E_1$ the estimate
$$\int |X|\,d\mu\leq \kappa\Rho^*(\mu)$$
holds, one could estimate
$$\rho(\kappa^{-1}|X|)\leq\frac 1{\kappa}\sup_{\mu\in E_1}\int |X|\,d\mu<\infty,$$
and thus $X\in\Lrho$. Thus, there must be a sequence $(\mu_n)_{n\in\Nat}\subseteq\dom(\Rho^*)$ such that $\Rho^*(\mu_n)>1$ and $\int |X|\,d\mu_n\geq 2^{2n}\Rho^*(\mu_n)$, $n\in\Nat$. 
We set $C:=\sum_{n=1}^\infty\frac 1{2^n\Rho^*(\mu_n)} \in (0,1)$, and 
$$\zeta:=\sum_{n=1}^\infty\frac 1{2^n\Rho^*(\mu_n)C}\mu_n.$$
As $\mu_n(\Omega)\leq\Rho(1)+\Rho^*(\mu_n)$, $\zeta(\Omega)$ is finite. Moreover, by $\sigma(\countable_\Prob,\Linfty)$-lower semicontinuity of $\Rho^*$, $\Rho^*(\zeta)\leq \frac 1C\sum_{n=1}^\infty2^{-n}=\frac 1C.$
Note that 
$$\int |X|\,d\zeta\geq\sum_{n=1}^\infty\frac{2^{2n}\Rho^*(\mu_n)}{2^n\Rho^*(\mu_n)C}=\frac 1C\sum_{n=1}^\infty 2^n=\infty,$$
Hence, for $\nu:=C\zeta+(1-C)\mu_0\in E_1$, where $\mu_0\in E_0$ is chosen arbitrarily, we also obtain $\int |X|\,d\nu=\infty$. This is the desired \textsc{contradiction}.

\smallskip
\noindent
(iii) follows from \cite[Proposition 4.10]{3}, and (iv) is an immediate consequence of the monotonicity of $\rho(|\cdot|)$.
\end{proof}

\noindent The proof of Proposition~\ref{thm:eta} will clarify the reason for introducing the norms $\Norm_{c,\mathcal R}$ instead of just $\Norm_{\mathcal R}$.

\begin{remark}\label{neededonce}\begin{enumerate}
\item[(i)] In the coherent case, we can infer from Proposition \ref{R12}(ii) that $\Norm_{c,\mathcal R}=c^{-1}\rho(|X|)=c^{-1}\sup_{\mu\in\dom(\Rho^*)}\|X\|_{L^1_\mu}$.
\item[(ii)] The Minkowski norm $\Norm_{c,\mathcal R}$ can be interpreted as a generalisation of the so-called \textit{Aumann-Serrano economic index of riskiness} (see \cite{14} and \cite[Example 3]{17}).
\item[(iii)] The Minkowski domain and similar spaces have appeared in \cite{KupSvi, 3, 4}. The definition of $\Lrho$ depends on the null sets of the probability measure $\Prob$ only, and thus is invariant under any choice of the underlying probability measure $\Prob\approx\dom(\Rho^*)$, in particular under changes of weak and strong reference probability models. 
\end{enumerate}
\end{remark}

\noindent The main purpose for introducing the Minkowski domain $\Lrho$ is to extend $\Rho$ to a larger domain than $\Linfty$ in a robust way in terms of the fundamentals, i.e.\ the risk measurement regime $\mathcal R=(\acc, \Scal,\price)$. 
There is a canonical candidate for this given by $\tilde{\mathcal R}:=(\tilde\acc, \mathcal S,\price)$ where  
\begin{equation}\label{tildeA}
     \tilde\acc:=\{X\in\Lrho\mid \forall \mu\in\dom(\Rho^*):~\int X\,d\mu\leq\Rho^*(\mu)\},
\end{equation}
so $\tilde\acc$ is given by lifting---and thus also preserving---the acceptability criteria $\int X\,d\mu\leq\Rho^*(\mu)$, $\mu\in\dom(\Rho^*)$, from $\Linfty$ to $\Lrho$. Indeed the following Theorem~\ref{lem:tildeA} shows that $\tilde{\mathcal R}$ is a risk measurement regime, and that the corresponding risk measure $\tilrho$ preserves the dual representation of $\Rho$. Dual approaches to extending convex functions are commonly used in the literature; see, e.g., \cite{7,3}.  Note that $\tilrho$ also preserves any functional form $\Rho$ may have, as for instance in the case of the entropic risk measure in Example~\ref{ex:entropic} below. 
 
\begin{theorem}\label{lem:tildeA}$\tilde{\mathcal R}:=(\tilde\acc, \mathcal S,\price)$ is a risk measurement regime on the Banach lattice $(\Lrho,\Norm_{\mathcal R})$. $\tilrho$ can be expressed as 
\begin{equation}\label{eq:g2}\tilrho(X)=\sup_{\mu\in \dom(\Rho^*)}\int X\, d\mu-\Rho^*(\mu), \quad X\in \Lrho,\end{equation}
where $\Rho^\ast$ is defined as in \eqref{eq:dualrep3}. Moreover, for every $\mu\in\dom(\Rho^*)$, the linear functional $\int\cdot\,d\mu$ is bounded on $(\Lrho,\Norm_{\mathcal R})$. \tn{A fortiori}, $\tilrho|_{\Linfty}=\Rho$. Moreover, $\tilrho$ is l.s.c.\ on $(\Lrho,\Norm_{\mathcal R})$, and satisfies 
\begin{equation}\label{approx}\tilrho(X)=\sup_{m\in\Nat}\tilrho(X\wedge m).\end{equation} 
\end{theorem}
\begin{proof}Note that for arbitrary $\mu\in\dom(\Rho^*)$ and all $X\neq 0$, we have
\begin{equation*}\label{domindual}\int \frac{|X|}{\|X\|_{\mathcal R}}d\mu=\sup_{\eps>0}\int \frac{|X|}{\|X\|_{\mathcal R}+\eps}d\mu\leq\sup_{\eps>0}\,\rho\left(\frac{|X|}{\|X\|_{\mathcal R}+\eps}\right)+\Rho^*(\mu)\leq 1+\Rho^*(\mu),\end{equation*}
hence $\int\cdot\,d\mu$ is a bounded linear functional on $\Lrho$. For arbitrary $X\in\Lrho$ and $\mu\in E_0$ we have 
\begin{center}$\sup\{\price(Z)\mid Z\in \Scal, X+Z\in \tilde\acc\}\leq \sup\{\price(Z)\mid Z\in \Scal, \price(Z) \leq  -\int X d\mu \}= -\int X d\mu<\infty,$\end{center}
where the finiteness of the bound is due to Proposition \ref{R12}(ii). Thus, $\tilde{\cal R}$ satisfies \eqref{regime} and is indeed a risk measurement regime, because $\tilde \acc$ is monotone by $\dom(\Rho^*)\subseteq(\countable_{\Prob})_+$, and convex as intersection of convex subsets of $\Lrho$. It is straightforward to show \eqref{eq:g2}, so $\tilrho$ is l.s.c.\ as pointwise supremum of a family of continuous functions. In order to prove \eqref{approx}, let $\mu\in\dom(\Rho^*)$ be arbitrary and note that by the Monotone Convergence Theorem and monotonicity of $\tilrho$, we have
\[\int X\,d\mu-\Rho^*(\mu)=\sup_{m\in\Nat}\int (X\wedge m)d\mu-\Rho^*(\mu)\leq\sup_{m\in\Nat}\tilrho(X\wedge m)\leq \tilrho(X).\]
Now take the supremum over $\mu\in\dom(\Rho^*)$ on the left-hand side.  
\end{proof}
\noindent Another way to extend $\Rho$ could be considering      
\begin{equation}\label{clA}
\overline\acc:=\cl_{\Norm_{\mathcal R}}(\acc). 
\end{equation}
and $\overline{\mathcal R}=(\overline\acc, \mathcal S, \price)$. We will discuss this approach in Remark~\ref{closureextension} where we show that $\overline{\mathcal R}$ is no risk measurement regime on $\Lrho$ in general, and that, where $\rho_{\overline{\mathcal R}}$ makes sense, it indeed equals $\tilrho$.
As announced in the introduction, we also consider the following extensions of $\Rho$ given by monotone approximation procedures:  
\begin{equation*}\label{eq:xi}
\xi(X):=\sup_{m\in\Nat}\inf_{n\in\Nat}\Rho((-n)\vee X\wedge m),\quad X\in\Lrho,\end{equation*}
and 
\begin{equation*}
\eta(X):=\inf_{n\in\Nat}\sup_{m\in\Nat}\Rho((-n)\vee X\wedge m),\quad X\in\Lrho.\end{equation*}
The question is under which conditions we have \begin{equation}\label{eq:eta:eq:tilrho}\tilrho(X)=\xi(X)=\eta(X)=\lim_{n\to\infty}\Rho((-n)\vee X\wedge n).\end{equation} Note that as a byproduct of \eqref{approx}, we obtain the estimate 
\begin{equation}\label{apriori}\tilrho\leq\xi\leq\eta\quad\tn{and}\quad\forall X\in\Lrho:~\tilrho(|X|)=\xi(|X|)=\eta(|X|)=\rho(|X|).\end{equation} The following Theorem~\ref{mainresult} shows that $\tilrho$ possesses some regularity in terms of monotone approximation in that always $\tilrho=\xi$.\\
\begin{theorem}\label{mainresult}For all $X\in\Lrho$ and all $U\in\Linfty$ we have 
\begin{equation}\label{xiistilde}\tilrho(X+U)=\sup_{m\in\Nat}\inf_{n\in\Nat}\Rho((-n)\vee X\wedge m +U).\end{equation}
\textnormal{A fortiori}, the equality $\tilrho=\xi$ holds, and $\tilrho$ can equivalently be interpreted as the risk measure associated to the risk measurement regime $\mathcal R_\xi:=(\mathcal A_\xi,\Scal,\price)$ on $(\Lrho,\Norm_{\cal R})$, where 
\[\acc_\xi:=\{X\in\Lrho\mid\xi(X)\leq 0\}=\{X\in\Lrho\mid\sup_{m\in\Nat}\inf_{n\in\Nat}\Rho((-n)\vee X\wedge m)\leq 0\}.\] 
\end{theorem}
\noindent For the sake of brevity, we shall in the remainder of our investigations often use the following piece of notation: for random variables $U,V\in(\Linfty)_+$ and $X\in\Lrho$, we set $X_U:=X\vee(-U)$ and $X^V:=X\wedge V$.
\begin{proof}
We show first that $\tilrho=\xi$ holds. Let $X\in\Lrho$, $m\in\Nat$ be fixed and $n\in\Nat$ be arbitrary. Let $\mu\in\dom(\Rho^*)$ be such that 
\[\tilrho(-X^-)-1\leq\Rho(X_n^m)-1\leq \int X^m_nd\mu-\Rho^*(\mu)\leq\int(X^+)^md\mu-\Rho^*(\mu).\]
Of course, the first and last inequalities in the latter estimate always hold by monotonicity. For $\eps>0$ arbitrary we can thus estimate
\begin{align*}\label{eq:controlsup}\begin{split}\Rho^*(\mu)-1&\leq \int (X^+)^md\mu-\tilrho(-X^-)=\frac 1 {1+\eps}\int (1+\eps)(X^+)^md\mu-\tilrho(-X^-)\\
&\leq \frac 1{1+\eps}\Rho((1+\eps)(X^+)^m)+\frac 1 {1+\eps}\Rho^*(\mu)-\tilrho(-X^-).\end{split}
\end{align*}
Rearranging this inequality, we obtain
\[\Rho^*(\mu)\leq\frac 1{\eps}\Rho((1+\eps)(X^+)^m)+\frac{1+\eps}\eps\left(1-\tilrho(-X^-)\right)=:c,\]
a bound which is independent of $n\in\Nat$. Since $E_c=\{\mu\in\countable_\Prob\mid \Rho^*(\mu)\leq c\}$ is $\sigma(\countable_\Prob,\Linfty)$-compact by Lemma \ref{lem:strongref}, we conclude for all $n\in\Nat$ that $\Rho(X^m_n)=\max_{\mu\in E_c}f(\mu,n)$, where the function $f$ is given by 
\[f: E_c\times\Nat\to\Reals,\quad f(\mu,n):=\int X^m_nd\mu-\Rho^*(\mu),\]
Our aim is to apply Fan's Minimax Theorem \cite[Theorem 2]{16} to the function $f$ in order to infer 
\begin{equation}\label{positivity}\xi(X^m)=\inf_n\max_{\mu\in E_c}f(\mu, n)=\max_{\mu\in E_c}\inf_{n\in\Nat}f(\mu, n)=\max_{\mu\in E_c}\inf_{n\in\Nat}\int X^m_nd\mu-\Rho^*(\mu).\end{equation}
To this end we have to check the following conditions:
\begin{itemize}
\item $E_c$ is a compact Hausdorff space when endowed with the relative $\sigma(\countable_\Prob, \Linfty)$-topology. This follows from continuity from above. 
\item $f$ is convex-like on $\Nat$ in that for all $n_1,n_2\in\Nat$ and all $0\leq t\leq 1$ there is a $n_0\in\Nat$ such that 
\[\forall\mu\in E_c:~f(\mu, n_0)\leq tf(\mu, n_1)+(1-t)f(\mu, n_2).\]
Indeed, choose $n_0:=\max\{n_1, n_2\}$ and note that 
\begin{align*}tf(\mu, n_1)+(1-t)f(\mu, n_2)&=t\int X^m_{n_1}d\mu+(1-t)\int X^m_{n_2}d\mu-\Rho^*(\mu)\\
&\geq (t+1-t)\int X^m_{n_0}d\mu-\Rho^*(\mu)=f(\mu, n_0).\end{align*}
\item $f$ is concave-like on $E_c$, which is defined analogous to convex-like. Indeed, let $\mu_1, \mu_2\in E_c$ and define $\mu_0=t\mu_1+(1-t)\mu_2\in E_c$ (by convexity of $E_c$). Then for all $n\in\Nat$, convexity of $\Rho^*$ implies
\begin{align*}tf(\mu_1,n)+(1-t)f(\mu_2, n)&=\int X^m_nd\mu_0-t\Rho^*(\mu_1)-(1-t)\Rho^*(\mu_2)\\
&\leq\int X^m_nd\mu_0-\Rho^*(\mu_0)=f(\mu_0,n).\end{align*}
\item For all $n\in\Nat$, the mapping $\mu\mapsto f(\mu,n)$ is upper semicontinuous. This follows from the continuity of $\mu\mapsto\int X^m_n d\mu$ and the lower semicontinuity of $\Rho^*$. 
\end{itemize}
From \eqref{positivity}, by the positivity of $\mu$ and, e.g., dominated convergence, 
\[\xi(X^m)=\max_{\mu\in E_c}\int X^md\mu-\Rho^*(\mu)\leq\tilrho(X^m),\]
and $\tilrho(X^m)=\xi(X^m)$ holds by \eqref{apriori}. Taking the limit $m\to\infty$, we obtain from the definition of $\xi$ and \eqref{approx} that $\tilrho(X)=\xi(X)$. \\
Now, let $X\in\Lrho$ and $U\in\Linfty$ be arbitrary and assume $m,n\geq u:=\|U\|_{\infty}$. We obtain  
\begin{align}\label{cutoff1}\begin{split}(X+U)_n&=(X+U)\Ind_{\{X\geq -U-n\}}-n\Ind_{\{X<-U-n\}}\\
&=X\Ind_{\{X\geq -U-n\}}-(n+U)\Ind_{\{X<-U-n\}}+U=X_{U+n}+U,\end{split}\end{align}
and in addition
\begin{align}\label{cutoff2}\begin{split}(X+U)^m&=(X+U)\Ind_{\{X\leq m-U\}}+m\Ind_{\{X>m-U\}}\\
&=X\Ind_{\{X\leq m-U\}}+(m-U)\Ind_{\{X>m-U\}}+U=X^{m-U}+U.\end{split}\end{align} 
From these two equations \eqref{cutoff1} and \eqref{cutoff2} we infer
\begin{align*}\xi(X+U)&=\sup_{m\geq u}\inf_{n\geq u}\Rho((X_{U+n}+U)^m)=\sup_{m\geq u}\inf_{n\geq u}\Rho(X^{m-U}_{U+n}+U).\end{align*}
This implies that 
\begin{align*}\sup_{m\in\Nat}\inf_{n\in\Nat}\Rho(X^m_n+U)&=\sup_{m\geq u}\inf_{n\geq u}\Rho(X^{m-U}_{U+n}+U)=\sup_{m\geq u}\inf_{n\geq u}\Rho((X+U)^m_n)\\
&=\xi(X+U)=\tilrho(X+U).\end{align*}
\eqref{xiistilde} is proved. $\xi=\tilrho$ being $\Scal$-additive, monotone, and proper, directly implies $\mathcal R_\xi$ is a risk measurement regime. The equality $\tilrho=\xi=\rho_{\mathcal R_\xi}$ obviously holds true. 
\end{proof}
\noindent Theorem~\ref{mainresult} appeared as \cite[Lemma 2.8]{4} in the context of law-invariant monetary risk measures. Our proof not only serves as an alternative to the one given in \cite{4}, relying irreducibly on law-invariance, but also generalises the result to a much wider class of risk measures.\\
In contrast to Theorem~\ref{mainresult}, we demonstrate in Example~\ref{ex:etaneqtilrho} that $\tilrho\neq\eta$ may happen. Before we study conditions under which $\tilrho$ displays regularity in the sense of \eqref{eq:eta:eq:tilrho}, we show the following properties of $\eta$: 
\begin{proposition}\label{thm:eta}
Define the acceptance set \[\acc_\eta:=\{X\in\Lrho\mid\inf_{n\in\Nat}\tilrho((-n)\vee X)\leq 0\}\subsetneq\Lrho.\]
Then $\eta$ is the risk measure associated to the risk measurement regime $\mathcal R_\eta:=(\acc_\eta,\Scal,\price)$. 
Moreover, \begin{equation}\label{approx:alt}\forall\,  X\in\Lrho:\quad \eta(X)=\inf_{n\in\Nat}\tilrho((-n)\vee X), \end{equation} and  $$\Gamma:=\{X\in \Lrho\mid \exists\, \eps>0:\, \rho((1+\eps) X^+)<\infty \} = \tn{int}\, \dom(\eta)\subseteq  \tn{int}\, \dom(\tilrho).$$ 
\end{proposition}
\begin{proof}
From \eqref{approx} and $\eta|_{\Linfty}=\Rho=\tilrho|_{\Linfty}$, we immediately obtain that for all $X\in\Lrho$ the equality $\eta(X)=\inf_{n\in\Nat}\tilrho((-n)\vee X)$ holds. \eqref{apriori} shows that $\acc_\eta\subsetneq\Lrho$ and that $\eta$ is a proper function. In order to prove the theorem, it suffices to check $\Scal$-additivity, convexity and monotonicity. 
Let $Z\in\Scal$ and $X\in\Lrho$. From the $\Scal$-additivity of $\tilrho$ and \eqref{cutoff1} we obtain, using the notational conventions introduced before the proof of Theorem~\ref{mainresult}, that 
$$\eta(X)=\inf_{n\geq\|Z\|_\infty}\tilrho(X_{Z+n}+Z)=\inf_{n\geq\|Z\|_\infty}\tilrho(X_{Z+n})+\price(Z)=\eta(X)+\price(Z).$$
For each $n\in\Nat$, $f_n(x):=(-n)\vee x$ is convex and monotone, thus $\eta=\lim_n\tilrho\circ f_n$ is convex and monotone. Next we show that $\Gamma\subseteq \tn{int}\, \dom(\eta)$. To this end we first show that  
$\mathbf{B}:=\bigcup_{c>0}\{Y\in\Lrho\mid\|Y\|_{c,\mathcal R}<1\}\subseteq \tn{int}\, \dom(\eta)$. Indeed for any $X$ with $\|X\|_{c,\mathcal R}<1$, there is $\lambda<1 $ such that $\rho(|X|/\lambda)\leq c$, and thus 
$$\eta(X)\leq \eta(|X|)= \rho(|X|) \leq \lambda \rho(|X|/\lambda)\leq \lambda c<\infty,$$
\noindent so $\mathbf{B}\subseteq \dom(\eta)$. Moreover, by definition $\mathbf{B}$ is open in $(\Lrho, \Norm_{\mathcal R})$. Now, let $X\in \Gamma$, and thus $X^+\in \mathbf{B}$. Hence, there is $\delta>0$ and a ball $B_\delta(0):=\{Y\in \Lrho\mid \|Y\|_{\mathcal R}<\delta\}$ such that $\{X^+\}+ B_\delta(0)\subseteq \dom(\eta)$. By monotonicity of $\eta$ it now follows that also $\{X\}+B_\delta(0)=\{X^+\}+B_\delta(0)- \{X^-\}\subseteq \dom(\eta)$, so $X\in \tn{int}\, \dom(\eta)$.\\
 In order to show $\Gamma\supseteq \tn{int}\, \dom(\eta)$ let $X\in \tn{int}\, \dom(\eta)$. Then there is $\eps>0$ such that $(1+2\eps)X\in \dom(\eta)$ and thus also $(1+\eps)X\in \dom(\eta)$, and by \eqref{approx:alt} there must  be $n\in \N$ such that $(1+2\eps)((-n)\vee X)\in \dom(\tilrho)$ and $(1+\eps)((-n)\vee X)\in \dom(\tilrho)$. Let  $X_n:= (-n)\vee X$ and $Y=(1+\eps)(X^-\wedge n)\in \Linfty$, so we have $(1+\eps)X^+=(1+\eps)X_n+ Y$. If $\delta>0$ satisfies $(1+\delta)(1+\eps)=1+2\eps$, convexity implies
\begin{align}
\rho((1+\eps)X^+)&= \tilrho((1+\eps)X_n+ Y)\;  = \; \tilrho\left(\frac{1+\delta}{1+\delta} (1+\eps)X_n +\frac{\delta(1+\delta)}{\delta(1+\delta)}Y\right) \nonumber \\ \label{eq:finite} &\leq \frac{1}{1+\delta} \tilrho((1+2\eps)X_n) +  \frac{\delta}{(1+\delta)} \tilrho\left( \frac{(1+\delta)}{\delta}Y\right)<\infty. 
\end{align} 
Hence,  $X\in \Gamma$. $\tn{int}\, \dom(\eta)\subseteq \tn{int}\, \dom(\tilrho)$ follows from $\tilrho\leq \eta$, see \eqref{apriori}. 
\end{proof}
\noindent The following Theorem~\ref{thm:tilrho:eta} states conditions under which \eqref{eq:eta:eq:tilrho} holds.  
\begin{theorem}\label{thm:tilrho:eta}
Let $X\in \Gamma$. Consider the following conditions:
\begin{itemize}
\item[(i)] there is $s>0$ such that for all $n\in \N$ we have 
\begin{center}$\tilrho((-n)\vee X)=\lim_{m\to \infty}\tilrho((-n)\vee X+s X^+\Ind_{\{X^+\geq m\}})$;\end{center}
\item[(ii)] there is $s>0$ such that $\eta(X)=\lim_{m\to \infty}\eta(X+s X^+\Ind_{\{X^+\geq m\}})$;
\item[(iii)] for all $n\in \N$ we have $\lim_{m\to \infty}\rho(n X\Ind_{\{X\geq m\}})=0$.
\end{itemize}
Any of the conditions (i)-(iii) implies \eqref{eq:eta:eq:tilrho}.
\end{theorem}
\noindent The set $\Gamma$ appears to be a set of reasonable risks in that they can at least be leveraged by a small amount and still remain hedgeable. Risks outside $\Gamma$ should probably not be considered by any sound agent. Note that the conditions (i)-(iii) are satisfied whenever  monotone or dominated convergence results can be applied to $\tilrho$, as is the case for many risk measures used in practice like the entropic risk measure in Example~\ref{ex:entropic} or Average Value at Risk based risk measures in Example~\ref{ex:avar}. The proof of Theorem~\ref{thm:tilrho:eta} is based on a study of subgradients of $\tilrho$ and $\eta$, respectively, and therefore postponed to the end of Section~\ref{sec:subgrad}. It turns out that the regularity condition  \eqref{eq:eta:eq:tilrho} is closely related to the existence of regular subgradients for $\eta$ and $\tilrho$.

\subsection{The structure of the Minkowski domain}\label{sec:structure}

In this section, we will decompose $\Lrho$ into parts with clear operational meanings. 
\begin{definition}We denote the closure of $\Linfty$ in $\Lrho$ by $\Mrho:=\tn{cl}_{\Norm_{\mathcal R}}(\Linfty)$, and define the \textsc{heart} of the Minkowski domain to be
\begin{center}$\heart:=\{X\in \Lrho\mid \rho(k|X|)<\infty~\tn{for all }k>0\}.$\end{center} 
\end{definition}

\noindent $\heart$, a concept which clearly adapts the idea of an Orlicz heart,\footnote{ For an introduction to Orlicz space theory we refer to \cite{Rao}.} is the set of risky positions which can be hedged at any quantity with finite cost. 
\begin{proposition}\label{R17} $\Mrho$ and $\heart$ are solid Banach sublattices of $\Lrho$ and $\Mrho\subseteq\heart$. 
Moreover, $\heart\subseteq\Gamma$, and both $\tilrho|_{\heart}$ and $\eta|_{\heart}$ are continuous.
\end{proposition}

\begin{proof}The first assertions are easily verified. Recall the set $\mathbf B$ from the proof of Proposition~\ref{thm:eta} for which we know that $\mathbf{B}\subseteq \Gamma$. For the inclusion $\heart\subseteq\mathbf B$, let $0\neq X\in \heart$ and note that $\rho(2|X|)<\infty$. The latter means 
$\|X\|_{c,\mathcal R}\leq\frac 1 2<1$ for some $c>0$, and thus $\heart\subseteq\mathbf B$. Finally, as $(\heart, \Norm_{\mathcal R})$ is a Banach lattice and both $\eta$ and $\tilrho$ are convex, monotone and finite-valued on $(\heart, \Norm_{\mathcal R})$, $\tilrho|_{\heart}$ and $\eta|_{\heart}$ are continuous according to Remark~\ref{shapiro}(v).
\end{proof}
\noindent From Proposition \ref{R17} we can derive the following characterisation of $\Mrho$, a result which can also be found as \cite[Lemma 3.3]{3}.
\begin{corollary}\label{Mrhocharac}
$\Mrho=\{X\in\Lrho\mid \forall\lambda>0:~\lim_{k\to\infty}\rho(\lambda|X|\Ind_{\{|X|\geq k\}})=0\}.$
\end{corollary}
\begin{proof}
Let $X\in\Mrho$ and $\lambda, \eps>0$ be arbitrary. Let $\delta>0$ be such that $\|Y\|_\mathcal R\leq \delta$, $Y\in\heart$, implies $\rho(|Y|)=\tilrho(|Y|)\leq \eps$. This is possible due to Proposition \ref{R17}. Choose now $Y\in \Linfty$ such that $\|\lambda(X-Y)\|_{\cal R}\leq \frac{\delta}2$ and $k\in\Nat$ such that $\|\lambda Y\Ind_{\{|X|\geq k\}}\|_{\cal R}\leq\frac{\delta}2$, the latter being due to continuity from above. Then $Z:=|X-Y|\Ind_{\{|X|\geq k\}}+|Y|\Ind_{\{|X|\geq k\}}$ satisfies $\|\lambda Z\|_{\cal R}\leq \delta$, 
and by monotonicity $\rho(\lambda|X|\Ind_{\{|X|\geq k\}})\leq\tilrho(\lambda Z)\leq \eps$. The converse inclusion above is obvious. 
\end{proof}

\noindent
As $\heart$ is closed, the set of directions along whose absolute value $\tilrho$ attains the value infinity is thus norm-open. In particular, we can only approximate such vectors with sequences of vectors along which $\tilrho$ behaves equally discontinuous, and limits of well-behaved financial positions are equally well-behaved. Hence shifting to $\Lrho$ yields a structure which conveniently separates regimes of ``good'' and ``bad'' risk behavior. 
In that respect consider the set $C^{\mathcal R}:=\dom(\tilrho)\backslash\heart\subseteq\Lrho$. $C^{\mathcal R}$ is the set of ``less bad'' positions, and shields $\heart$ from the financial positions that carry infinite risk. It has a nice interpretation in terms of liquidity risk in the sense of Lacker \cite{15}. In that paper the author considers liquidity risk profiles, i.e.\ curves of the form $\tilrho(t X)_{t\geq 0}$ capturing how risk scales when increasing the leverage. $C^{\mathcal R}$ consists of financial positions $X$ such that the liquidity risk profiles of $X^+$ or $X^-$ breach the infinite risk regimes. Whereas an agent could at least hypothetically hedge any position in $\heart$ at finite cost, no matter what the leverage, she has to be very careful in the case of elements in $C^{\mathcal R}$ that have finite risk themselves but which produce potentially completely non-hedgeable losses under incautious scaling. \\
Recalling that for any $X\in \Lrho$ there is $\lambda>0$ such that $\rho(|X|/\lambda)<\infty$, we obtain that $C^{\mathcal R}=\emptyset$ if and only if $\heart=\Lrho$, and $\tilrho$ is continuous. 
Moreover, if $\heart\subsetneq\Lrho$, both $\heart$ and $\Mrho$ 
are nowhere dense (as true subspaces of $\Lrho$) and---by Baire's Theorem---$C^{\mathcal R}\cup\{\tilrho=\infty\}$ is a dense open set.\\
Note that the inclusions $\Mrho\subseteq \heart\subseteq\Lrho$ can all be strict, as is illustrated by Example~\ref{inclusions}.

\begin{remark}\label{closureextension} Having introduced $\Mrho$ we can now discuss the extension given by the norm closure operation \eqref{clA}. Seen as a subset of $\Lrho$, $\overline\acc$ is unfortunately not an acceptance set in the sense of Definition~\ref{R1}, since $X\leq Y$ and $Y\in\overline\acc$ does not necessarily imply $X\in\overline\acc$, so the monotonicity property is violated. However, one can show that 
$\overline{\mathcal R}:=(\overline\acc, \Scal,\price)$ is a risk measurement regime on the Banach lattice $\Mrho$. 
By Proposition \ref{R17} it follows that  $\rho_{\overline{\cal A}}(X)=\tilrho(X)=\eta(X)$ for all $X\in\Mrho$, and $\rho_{\overline{\cal A}}$ is continuous on $\Mrho$. 
\end{remark} 

\subsection{The dual of the Minkowski domain}\label{sec:dual}

In this short interlude we discuss a few properties of the norm dual $(\Dualrho,\Norm_{\mathcal R*})$ of $(\Lrho,\Norm_{\mathcal R})$, the space of continuous linear functionals on the Minkowski domain, which will be essential when we study subgradients in Section~\ref{sec:subgrad}. 
\begin{theorem}\label{R22}
$\Dualrho$ is the direct sum of two subspaces $CA$ and $PA$, i.e. 
\begin{center}$\Dualrho=CA\oplus PA.$\end{center}
Elements in $CA$ have the shape $X\mapsto \int X\,d\mu$ for a unique $\mu\in\countable_\Prob$. $\lambda\in PA$ are characterised by $\lambda|_{\Mrho}=0$. For $\ell=\mu\oplus\lambda$,\footnote{ We shall stick to the abuse of notation of identifying functionals in $CA$ with the unique measure $\mu\in\countable_\Prob$ in their integral representation.} $\mu$ is the \textsc{regular part} of $\ell$ and $\lambda$ the \textsc{singular part}. Moreover, $\Linfty$ can be identified with a subspace of $\Dualrho$. 
\end{theorem}
\begin{proof}
Let $\ell\in\Dualrho$ and consider the additive set function $\mu=\mu_\ell:\mathcal F\to\Reals$, $\mu(A):=\ell(\Ind_A)$. It is straightforward to prove that $\mu\in\bounded_{\Prob}$ and that it is unique, given $\ell$. Let now $(A_n)_{n\in\Nat}\subseteq\mathcal F$ be a vanishing sequence of sets. For all $\lambda>0$ continuity from above implies 
\begin{center}$\lim_{n\to\infty}\rho(\lambda^{-1}\Ind_{A_n})=\rho(0)=0,$\end{center}
which reads as $\lim_n\|\Ind_{A_n}\|_{\mathcal R}=0$ and thus $\lim_n\mu(A_n)=\lim_n\ell(\Ind_{A_n})=0$. Hence $\mu\in\countable_\Prob$.\\
We will now show that the linear functional $X\mapsto\int X\,d\mu$ is bounded. To this end, note first that by its definition, $\ell(X)=\int X\,d\mu_\ell$ holds for all $X\in\Linfty$. Moreover, by \cite[Theorem 7.46]{Ali}, $\Dualrho$ is a Banach lattice in its own right, and the mapping $\ell\mapsto\mu_\ell$ is positive and linear in $\ell$, hence it suffices to show $X\mapsto \int X\,d\mu\in\Dualrho$ is bounded for $\mu=\mu_\ell\in(\countable_\Prob)_+$, $\ell\in\Dualrho_+$. Let $X\in\Lrho_+$ be arbitrary.
\[\int (X\wedge n)d\mu=\left|\ell(X\wedge n)\right|\leq\|\ell\|_{\mathcal R *}\|(X\wedge n)\|_{\mathcal R}\leq\|\ell\|_{\mathcal R*}\|X\|_{\mathcal R},\]
where the last inequality follows from Proposition \ref{R12}(iv). We apply the Monotone Convergence Theorem and obtain $\int X\,d\mu\leq\|\ell\|_{\mathcal R *}\|X\|_{\cal R}$. 
For a general $X\in\Lrho$, we get
\begin{align*}\left|\int X\,d\mu\right|&\leq\int |X|d\mu\leq\|\ell\|_{\mathcal R *}\| |X| \|_{\cal R}=\|\ell\|_{\mathcal R *}\|X\|_{\cal R}.
\end{align*}
$X\mapsto\int X\,d\mu\in\Dualrho$ follows, and from $\Linfty$ being dense in $\Mrho$, $\ell|_{\Mrho}=\int \cdot\,d\mu|_{\Mrho}$ has to hold. Let $CA:=\{\int\cdot\,d\mu_{\ell}\mid \ell\in\Dualrho\}$, which is a subspace of $\Dualrho$. For $\ell\in\Dualrho$, let $\lambda:=\ell-\int\cdot\,d\mu\in\Dualrho$, which satisfies $\lambda|_{\Mrho}=0$. Clearly, $\ell=\int \cdot\,d\mu+\lambda$ is a unique decomposition of $\ell$ as a sum of elements in $CA$ and $PA$.\\
If $Z\in\Linfty$, the inclusion $\Lrho\subseteq L^1_{\Prob}$, H\"older's inequality and \eqref{16} yield $\Lrho\ni X\mapsto\Erw_\Prob[ZX]$ is well-defined and continuous, i.e. $\Erw_\Prob[Z\cdot]\in\Dualrho$.
\end{proof}

\noindent $CA$ stands for ``countably additive'', $PA$ for ``purely additive''.\ One can show that $CA$ is a closed subspace of $\Dualrho$. The following corollary is a direct consequence of Theorem~\ref{R22}.
\begin{corollary}\label{singularities}For all $\lambda\in PA$, $X\in\Lrho$ and $r>0$, we have the identity
\[\lambda(X)=\lambda(X\Ind_{\{|X|\geq r\}}).\]
Moreover, if $\ell=\mu\oplus\lambda\in\Dualrho$, $\lim_{r\to\infty}\ell(X\Ind_{\{|X|\geq r\}})=\lambda(X)$ holds for all $X\in\Lrho$.
\end{corollary}
\noindent Theorem~\ref{R22} implies another characterisation of $\tilrho$. 
\begin{corollary}\label{EQUALITY}
Consider the following two classes of extensions of $\Rho$ to $\Lrho$:
\begin{center}$\mathcal E_1=\{g:\Lrho\to(-\infty,\infty]\mid g$ convex, $\sigma(\Lrho,CA)$-l.s.c., $g|_{\Linfty}=\Rho\}$,\\
$\mathcal E_2:=\{g:\Lrho\to(-\infty,\infty]\mid g$ monotone, $g=\sup_{m\in\Nat}g(\cdot\wedge m),\,g|_{\Linfty}=\Rho\}$.\end{center}
Then $\tilrho$ is maximal both in $\mathcal E_1$ and $\mathcal E_2$, i.e. $g\in\mathcal E_i$ implies $g\leq \tilrho$. 
\end{corollary}
\begin{proof}
First assume $g\in\mathcal E_1$. By the Fenchel-Moreau Theorem (c.f.~\cite[Proposition 4.1]{Ekeland}) $g$ has a dual representation
\[g(X)=\sup_{\mu\in CA}\int X\,d\mu-g^*(\mu),\quad X\in\Lrho,\]
where $g^*(\mu)=\sup_{X\in\Lrho}\int X\,d\mu-g(X)$. By $g|_{\Linfty}=\Rho$, we have $\dom(g^*)\subseteq \dom(\Rho^*)$ and $g^*(\mu)\geq \Rho^*(\mu)$ for all $\mu\in \dom(\Rho^*)$.  Hence, for $X\in\Lrho$ arbitrary, we have 
\[g(X)\leq\sup_{\mu\in\dom(\Rho^*)}\int X\,d\mu-g^*(\mu)\leq \sup_{\mu\in\dom(\Rho^*)}\int X\,d\mu-\Rho^*(\mu)=\tilrho(X).\]
For the second claim, let $g\in\mathcal E_2$ and let $X\in\Lrho$ be arbitrary. Monotonicity of $g$ allows for the following estimate:
\[g(X)=\sup_{m\in\Nat}g(X\wedge m)\leq\sup_{m\in\Nat}\inf_{n\in\Nat}\underbrace{g((-n)\vee X\wedge m)}_{=\Rho((-n)\vee X\wedge m)}=\xi(X)=\tilrho(X).\]
\end{proof}

\subsection{Subgradients over the Minkowski domain}\label{sec:subgrad}
In this section we will study subgradients of $\tilrho$ and $\eta$, and how to ensure that subgradients correspond to measures on $(\Omega, \Fcal)$. Given Theorem~\ref{R22}, it does not seem surprising that this is not always the case. The reason for also considering subgradients of $\eta$ is that existence of regular subgradients of $\eta$ and $\tilrho$ is closely related to the question \eqref{eq:eta:eq:tilrho}, and the developed results pave the way for the proof of Theorem~\ref{thm:tilrho:eta}.
\begin{definition}
Let $(\mathcal X,\tau)$ be a topological vector space with dual space $\Xcal^*$. Given a proper convex function $f:\Xcal\to(-\infty,\infty]$, the \textsc{subgradient} of $f$ at $X\in \Xcal$ is the set
\begin{align*}\partial f(X)&:=\{\ell\in\Xcal^*\mid \forall Y\in \Xcal: f(Y)\geq  f(X)+\ell(Y-X)\}\\ & =\{\ell\in\Xcal^*\mid f(X)=\ell(X)-f^*(\ell)\},\end{align*}  where $f^\ast(\ell):=\sup_{X\in \Xcal} \ell(X)-f(X)$, $\ell \in \Xcal^\ast$.
\end{definition}
\noindent Note that if a convex function $f:\Lrho\to(-\infty,\infty]$ is additionally monotone and $\Scal$-additive, its subgradients will be positive functionals in $\Dualrho_+$ that agree with $\price$ on $\Scal$. \\
In the study of risk measures subgradients play an important role, for instance as pricing rules in equilibria. The following easy example serves as an economic motivation.
\begin{example}[Optimal investment]
For some capital constraint $c>0$ and some linear pricing rule $\ell\in\Dualrho_+$ consider the following optimisation problem: 
\[(\ast)\quad \tilrho(Y)\to\min, \quad \mbox{over all} \; Y\in\Lrho\tn{ with }\ell(-Y)\leq c.\]
In order to solve this, by monotonicity, we can without loss of generality focus on $Y$ satisfying $\ell(-Y)=c$. If $X\in\Lrho$ satisfies $\ell\in \partial\tilrho(X)$ and $\ell(-X)=c$, then $X$ solves ($\ast$). Indeed for all $Y\in\ell^{-1}(\{-c\})$, we have
\[\tilrho(X)=\tilrho(X)+\ell(Y-X) -\ell(Y)-c\leq\tilrho(Y)+\ell(-Y)-c=\tilrho(Y).\]  
\end{example}

\noindent An important feature of the space $(\Lrho, \Norm_{\cal R})$ is that $\dom(\tilrho)$ possesses a particularly rich interior, see Proposition~\ref{thm:eta}. Thus we have the following result: 
\begin{theorem}\label{thm:subdiff:general}
Suppose $X\in \tn{int}\,\dom(\tilrho)$, so in particular if $X\in\Gamma$, then $\partial\tilrho(X)\neq \emptyset$. Also $\partial\eta(Y)\neq\emptyset$ whenever $Y\in\Gamma$. 
\end{theorem}
\begin{proof}
It is well-known that a convex, proper and monotone function $f$ on a Banach lattice is subdifferentiable at every point in $\tn{int}\, \dom (f)$, see \cite[Proposition 1]{Shapiro}. The claim thus follows from Theorem~\ref{lem:tildeA} and Proposition~\ref{thm:eta}.
\end{proof}

\noindent 
We devote the remainder of this subsection to the question under which conditions $\partial\tilrho(X)$ will contain regular (that is $\sigma$-additive) elements.\footnote{ There are immediate---however very strong---sufficient conditions for this to happen, e.g. $\Dualrho\subseteq \countable_\Prob$, which is the case if and only if $\Mrho=\Lrho$, or continuity of $\tilrho$ with respect to the $\sigma(\Lrho,CA)$-topology.} To this end, note that by \eqref{apriori} we have that $\eta^*\leq\tilrho^*$, which implies $\dom(\tilrho^*)\subseteq\dom(\eta^*)$. Moreover, $CA\cap \dom (\tilrho^*)\subseteq CA\cap\dom(\eta^*)\subseteq\dom(\Rho^*)$, so regular subgradients of $\tilrho$ and $\eta$ are necessarily in $\dom(\Rho^*)$. Indeed, if $\mu\in CA\cap \dom (\eta^*)$, then $$\Rho^*(\mu)=\sup_{Y\in\Linfty}\int Y\,d\mu-\eta(Y)\leq \sup_{Y\in\Lrho}\int Y\,d\mu-\eta(Y)= \eta^*(\mu)<\infty.$$
Conversely, for all $Y\in\dom(\eta)$ the definition of $\eta$ and $CA\subseteq\countable_\Prob$ shows for $\mu\in CA$
\begin{align}\label{eq:dualsequal}\begin{split}\int Y\,d\mu-\eta(Y)&=\lim_{n\to\infty}\lim_{m\to\infty}\int Y^m_nd\mu-\Rho(Y^m_n)\\
&\leq\sup_{U\in\Linfty}\int U\,d\mu-\Rho(U)=\Rho^*(\mu).\end{split}\end{align}
This shows that $\eta^*(\mu)=\Rho^*(\mu)$, which provides a first step towards the proof of Theorem~\ref{thm:tilrho:eta}:
\begin{lemma}\label{lem:eta=tilrho} Let $X\in\dom(\eta)$  and suppose that $\mu\oplus \lambda \in \partial\eta(X)$, where $\mu \in CA$ and $\lambda\in PA$. Then $\lambda(X^-)=0$. If, moreover, 
$\lambda=0$, i.e.\ $\mu \in \partial\eta(X)$, then $\eta(X)=\tilrho(X)$. 
\end{lemma}
\begin{proof} Let $\mu\oplus \lambda \in \partial\eta(X)$. Define $\tilde \lambda$ by $\tilde \lambda(Y)=\lambda(Y\Ind_{\{X\geq 0\}})$, $Y\in \Lrho$. One verifies that $\tilde\lambda\in (\Lrho)^\ast$. Also we have \begin{eqnarray*}\eta^\ast(\mu \oplus \tilde\lambda) &= &\sup_{Y\in \Lrho} \int Y\, d\mu + \lambda (Y\Ind_{\{X\geq 0\}})-\eta (Y)\\ &\leq & \sup_{Y\in \Lrho}\lim_{n\to \infty} \int (-n)\vee Y\, d\mu + \lambda (Y^+)-\eta ((-n)\vee Y)\\  & \leq &  \limsup_{n\to \infty}\sup_{Y\in \Lrho} \int (-n)\vee Y\, d\mu + \lambda ((-n)\vee Y)-\eta ((-n)\vee Y)\\ & \leq & \eta^\ast(\mu\oplus \lambda),\end{eqnarray*} where we used monotonicity  of $\lambda$. Hence, $$\eta(X)=\int X\, d\mu + \lambda (X)-\eta^\ast(\mu\oplus\lambda) \leq \int X\, d\mu + \tilde \lambda (X)-\eta^\ast(\mu\oplus\tilde \lambda)\leq \eta(X),$$ and the first inequality would be strict if $\lambda(X^-)>0$. Thus $\lambda(X^-)=0$ follows. For the last assertion, suppose that $\mu\in \partial\eta(X)\cap CA$. The observations preceding the lemma  and \eqref{apriori}~show
\[\eta(X)=\int X\,d\mu-\Rho^*(\mu)\leq\tilrho(X)\leq \eta(X).\]
\end{proof}
\noindent Consequently, if $\partial \eta(X)\neq \emptyset$, so for instance for $X\in \Gamma$, then $\eta$ may display a ``jump'' $\lambda(X^+)$ produced by the unbounded risk $X^+$. If that jump is not present, then $\eta(X)=\tilrho(X)$. 
In the following we will introduce a weak local continuity assumption, \textit{tail continuity}, which quantifies which tails are not too fat to lead to such jumps. In \cite{4}, a version of it is studied for law-invariant monetary risk measures.

\begin{definition} Let $f:\Lrho\to(-\infty,\infty]$ be monotone and proper, and let $X\in\dom(f)$. We call $f$ \textsc{tail continuous} at $X$ along $Y\in\Lrho$ if $X+Y^+\in\dom(f)$ and  
\[f(X)=\lim_{r\rightarrow\infty}f\left(X+Y\Ind_{\{Y\geq r\}}\right)\]
holds. $\mathcal T^f_X$ denotes the set of tails $Y$ along which $f$ is tail continuous at $X$. With a slight abuse of language, we call $f$ tail continuous at $X$ if $\mathcal T^f_X=\{Y\in\Lrho\mid X+Y^+\in\dom(f)\}$. 
\end{definition}
\noindent Note that $\mathcal T^f_X$ is monotone in that $Y_1\leq Y_2$ $\Prob$-a.s. and $Y_2\in\mathcal T^f_X$ implies $Y_1\in\mathcal T^f_X$. 
The next proposition shows that sufficient tail continuity can eliminate non-$\sigma$-additive elements in the subgradient.  We prove this for general monotone functions $f$, but we clearly have $f=\tilrho$ or $f=\eta$ in mind.
\begin{proposition}\label{prop:subgrad} Let $f:\Lrho\to(-\infty,\infty]$ be proper, monotone, and convex, and let $X\in \dom(f)$. Suppose that $\{sY\mid s\geq 0, Y\in \mathcal T^{f}_X\}$ is norm-dense (or equivalently $\mathcal T^{f}_X$ separates the points of $\Dualrho$). Then $\partial f(X)\subseteq CA$. In particular, if $f$ is tail continuous at $X\in\tn{int}\,\dom(f)$, then $\partial f(X)\subseteq CA$.
\end{proposition}
\begin{proof}
Let $\ell=\mu\oplus\lambda\in\partial f(X)$. Assume $\lambda\neq 0$. The density assumption and monotonicity allows to pick  $Y\in \mathcal T_X^{f}$, $Y\geq 0$, such that $\lambda(Y)>0$. Corollary~\ref{singularities} and $\ell$ being a subgradient together with tail continuity along $Y$ yield the \textsc{contradiction}
\begin{align*}f(X)&<f(X)+\lambda(Y)=\lim_{r\rightarrow\infty}f(X)+\lambda(Y\Ind_{\{Y\geq r\}})\\ 
&=\lim_{r\to\infty}f(X)+\ell(Y\Ind_{\{Y\geq r\}}) = \lim_{r\to\infty}\ell(X)-f^\ast(\ell)+\ell(Y\Ind_{\{Y\geq r\}})\\ & = \lim_{r\to\infty}\ell(X+Y\Ind_{\{Y\geq r\}})-f^\ast(\ell)\leq\liminf_{r\rightarrow\infty}f(X+Y\Ind_{\{Y\geq r\}})=f(X).
\end{align*}
\end{proof}

\noindent Unfortunately, in general we only have tail continuity along $\Mrho$, as is shown in the following Lemma~\ref{lem:1}. As we have already observed, if $\Lrho=\Mrho$, then $\Dualrho=CA$ and therefore trivially $\partial\tilrho(X)\subseteq CA$, so just knowing tail continuity along $\Mrho$ is not sufficient for the existence of countably additive subgradients in non-trivial cases.

\begin{lemma}\label{lem:1} Let $f:\Lrho\to(-\infty,\infty]$ be proper, monotone, and convex such that $\Linfty\subseteq \dom(f)$. If $X\in\tn{int}\,\dom(f)$, then $\Mrho_+- \Lrho_+\subseteq\mathcal T_X^{f}$.
\end{lemma}

\begin{proof} $\mathcal T_{X}^{f}$ is monotone, hence it suffices to consider $Y\in \Mrho_+$. The condition $X\in\tn{int}\,\dom(f)$ guarantees $X+Y\in\dom(f)$ as in \eqref{eq:finite}.
From Corollary \ref{Mrhocharac} we obtain $\lim_n\|Y\Ind_{\{Y\geq n\}}\|_{\mathcal R}=0$, hence $X+Y\Ind_{\{Y\geq n\}}\in\tn{int}\,\dom(f)$ for all $n$ large enough. The desired tail continuity follows from the continuity of $f|_{\tn{int}\,\dom(f)}$ (see Remark~\ref{shapiro}(v)).  
\end{proof}

\noindent While in Proposition \ref{prop:subgrad} we gave a condition under which the subgradient contains regular dual elements only, we will now turn to conditions guaranteeing the existence of at least one regular element in the subgradient, namely by means of projection. 
\begin{proposition}\label{subgradproj1}
Let $X\in\dom(\tilrho)$ and  $\ell=\mu\oplus\lambda\in\partial\tilrho(X)$. Then also $\mu\in\partial\tilrho(X)$ whenever $\mu$ satisfies $\int X\,d\mu\geq\ell(X)$. Similarly, if $X\in\dom(\eta)$ and  $\ell=\mu\oplus\lambda\in\partial\eta(X)$, then $\mu\in\partial\eta(X)$ whenever $\mu$ satisfies $\int X\,d\mu\geq\ell(X)$. In particular, the assumption $\int X\,d\mu\geq\ell(X)$ is met if $X\in \Mrho_+-\Lrho_+$.
\end{proposition}

\begin{proof}
By the same argument employed in \eqref{eq:dualsequal} and the equality $\tilrho=\xi$, we obtain that $\eta^\ast(\mu)=\tilrho^*(\mu)=\sup_{U\in\Linfty}\int U\,d\mu-\Rho(U)=\Rho^*(\mu)$ holds for all $\mu\in CA$. From this and $\ell|_{\Mrho}=\int \cdot\,d\mu$ we infer $\tilrho^*(\mu)\leq\tilrho^*(\ell)$, and $\eta^*(\mu)\leq\eta^*(\ell)$. The assumption $\int X\,d\mu\geq\ell(X)$ and $\ell\in \partial \tilrho(X)$ imply 
\[\tilrho(X)\geq\int X\,d\mu-\tilrho^*(\mu)\geq \ell(X)-\tilrho^*(\ell)=\tilrho(X).\] The assertion for $\eta$ follows in the same way.
\end{proof}
\begin{remark}
In the situation of Proposition \ref{subgradproj1}, as $\int X\,d\mu-\tilrho^*(\mu)=\ell(X)-\tilrho^*(\ell)$, $\tilrho^*(\mu)\leq\tilrho^*(\ell)$, and $\int X\,d\mu\geq\ell(X)$, we in fact obtain $\int X\,d\mu=\ell(X)$ and $\tilrho^*(\ell)=\tilrho^*(\mu)$. In other words, singularities in the subgradient cannot be excluded, but they are redundant for $X$.
\end{remark}
\noindent The following proposition establishes a handy criterion for $\int X\,d\mu\geq \ell(X)$. 

\begin{proposition}\label{subgradproj3}
Suppose that $f:\Lrho\to(-\infty,\infty]$ is monotone, proper and convex, and that $\ell=\mu\oplus\lambda\in\partial f(X)$. Then $\int X\,d\mu\geq\ell(X)$ whenever $sX^+\in\mathcal T_X^f$ for some $s>0$. 
\end{proposition}
\begin{proof}
Suppose that $\lambda(X^+)=:\delta>0$. By monotonicity one obtains for all $n\in\Nat$
\begin{align*}\delta&=\lambda(X^+)=\lambda(X^+\Ind_{\{X^+\geq n\}})\leq\ell(X^+\Ind_{\{X^+\geq n\}}).
\end{align*}
Define $X_n=X+s X^+\Ind_{\{X^+\geq n\}}\geq X$, $n\in\Nat$, where $s>0$ is chosen like in the assumption of the proposition. We estimate 
\begin{align*}\ell(X)-f(X)&=f^*(\ell)\geq\ell(X_n)-f(X_n)=\ell(X)+s\ell(X^+\Ind_{\{X^+\geq n\}})-f(X_n)\\
&\geq\ell(X)+s \delta-f(X_n).
\end{align*}
Consequently, we arrive at the \textsc{contradiction} $0=\lim_{n\rightarrow\infty}f(X_n)-f(X)\geq s\delta$.
Hence, $\lambda(X^+)=0$, and thus $\int X\,d\mu\geq \ell(X).$ 
\end{proof}

\noindent We now have the tools at hand to provide the proof of Theorem \ref{thm:tilrho:eta}.

\begin{proof}[Proof of Theorem~\ref{thm:tilrho:eta}]
Note that condition (i) implies (ii), and suppose that one of them holds. $X\in \Gamma$ implies that  $\partial\eta(X)\neq \emptyset$ (Theorem~\ref{thm:subdiff:general}), and Propositions~\ref{subgradproj1} and \ref{subgradproj3} in conjunction with the second part of Lemma~\ref{lem:eta=tilrho} do the rest. \\
Condition~(iii) is equivalent to $X^+\in \Mrho$ by Corollary~\ref{Mrhocharac}. Hence 
Proposition~\ref{subgradproj1} applies, and Lemma~\ref{lem:eta=tilrho} yields the assertion.
\end{proof}

\section{Examples}\label{sec:examples}

\begin{example}\label{ex:counterex}
Consider $(\Omega,\Fcal)$ to be the natural numbers endowed with their power set. Let $\zeta\in\countable_+$ be defined by the discrete density $(2^{-\om})_{\om\in\Nat}$ and let $\nu\in\bounded_+$ be the purely finitely additive measure on $(\Omega,\Fcal)$ arising from the Banach-Mazur limit (c.f. \cite[Definition 15.46]{Ali}); the reader should keep in mind that $\nu(F)=0$ for all finite sets $F\subseteq\Nat$. Moreover, for $\lambda \in [0,1]$ we set $\mu_\lambda=(1-\lambda)\zeta+\lambda\nu$ and define the closed acceptance set
$$\acc:=\left\{X\in\Bd\Big| \forall\lambda\in[0,1]:~\int X\,d\mu_\lambda\leq \lambda\right\}. $$
Clearly, $\mathcal B(\acc)\backslash\countable\neq \emptyset$. Now let $\emptyset\neq A\subset\Nat$ be any finite subset, $\Scal=\{U(\alpha,\beta):=\alpha\Ind_A+\beta\mid \alpha,\beta\in\Reals\}$, $\price(U(\alpha,\beta))=\int U(\alpha,\beta)\,d\zeta$. We first show that $\mathcal R:=(\acc,\Scal,\price)$ is a risk measurement regime and $\Rho$ is finite. To this end, note first that, for arbitrary $X\in\Bd$, $X+U(\alpha,\beta)\in\acc$ implies 
$$0\geq\int(X+U(\alpha,\beta))d\mu_0=\int X\,d\zeta+\price(U(\alpha,\beta)), $$
hence $\price(U(\alpha,\beta))\leq -\int X\,d\zeta<\infty$, and $\mathcal R$ is a risk measurement regime. Moreover, for any $k\in\Nat$ we have that $k-U(0,k)\in\acc$, which means that $\Rho(k)\leq k$. By monotonicity, $\Rho$ does not attain the value $+\infty$.\\ 
Next we prove that $\Rho$ is continuous from above even though $\mathcal B(\acc)\setminus \countable\neq \emptyset$. We proceed in three steps.\\
\textit{Step 1:} $\sigma_\acc(\mu_\lambda)=\lambda$. Indeed, let $A_n:=\{n,n+1,...\}$ and note that
$$\mu_\lambda(A_n)=(1-\lambda)\sum_{i=n}^\infty 2^{-i}+\lambda.$$
Hence $Y_n:=\Ind_{A_n}-\sum_{i=n}^\infty2^{-i}\in\acc$, and 
$$\sigma_\acc(\mu_\lambda)\geq \lim_{n\to\infty}\int Y_n\,d\mu_\lambda=\lim_{n\to\infty}-\lambda\sum_{i=n}^\infty 2^{-i}+\lambda=\lambda.$$
The converse inequality $\sigma_\acc(\mu_\lambda)\leq \lambda$ is due to the definition of $\acc$.\\ 
\textit{Step 2:} $\mathcal B(\acc)=\tn{cone}(\{\zeta,\nu\})$, where cone$(E)$ refers to the smallest convex and pointed cone containing $E\subseteq\bounded$ and $0$. The inclusion $\mathcal B(\acc)\supseteq\tn{cone}(\{\zeta,\nu\})$ is clear, for the other one note that $\mathcal B(\acc)=\tn{cone}(\mathcal B(\acc)_1)$ always holds, where $\mathcal B(\acc)_1:=\{\frac 1{\mu(\Omega)}\mu\mid 0\neq\mu\in\mathcal B(\acc)\}$. Assume we can find $\mu\in\mathcal B(\acc)_1\backslash co(\{\zeta,\nu\})$, where $co(\{\zeta,\nu\})$ denotes the convex hull of $\zeta$ and $\nu$. As $ co(\{\zeta,\nu\})$ is $\sigma(\bounded,\Bd)$-compact and convex, by means of separation we can find a $Y\in\Bd$ such that 
\[\max_{\lambda\in[0,1]}\left(\int Y\,d\mu_\lambda-\lambda\right)\leq \max_{\lambda\in[0,1]}\int Y\,d\mu_{\lambda}=0<\int Y\,d\mu.\]
The same holds true when $Y$ is replaced by $tY$, $t>0$. Thus $\{tY\mid t>0\}\subseteq\acc$, and 
$$\sigma_\acc(\mu)\geq\sup_{t>0}\int tY\,d\mu=\infty.$$
We conclude that $\mathcal B(\acc)_1=co(\{\zeta,\nu\})$ and thus $\mathcal B(\acc)=\tn{cone}(\{\zeta,\nu\})$.\\
\textit{Step 3:} $\mathcal E_{\price}\cap \mathcal B(\acc)=\{\zeta\}$ and therefore $\Rho(X)=\int X\,d\zeta$ by Proposition \ref{ancillary}(i), which is continuous from above. Indeed, $\mu\in\mathcal E_{\price}\cap \mathcal B(\acc)$ only if $\mu\in\mathcal B(\acc)_1$, therefore by Step 2 we can assume $\mu_\lambda\in\mathcal E_\price\cap\mathcal B(\acc)$ for some $\lambda\in [0,1]$. We reformulate the condition as for all $\alpha,\beta\in\Reals$ it has to hold
$$\alpha\zeta(A)+\beta=(1-\lambda)(\alpha\zeta(A)+\beta)+\lambda\beta=(1-\lambda)\alpha\zeta(A)+\beta,$$
which is the case if and only if $\lambda=0$. 
\end{example}

\begin{example}[$\tilrho\neq \eta$]\label{ex:etaneqtilrho}
Let $(\Omega,\Fcal,\Prob)$ be the integers $\mathbb Z$ endowed with their power set and a probability measure specified below. Let 
$$\Quot_k:=\frac 1{2k}(\delta_k+\delta_{-k})+\left(1-\frac 1 k\right)\delta_0,\quad k\in\Nat,$$ and define $\Prob:=\sum_{k\in\Nat}2^{-k}\Quot_k.$
It is straightforward to check that  
\[\acc:=\{X\in \Linfty\mid\forall k\in\Nat:~\Erw_{\Quot_k}[X]\leq 0\},\quad\Scal=\Reals,\quad\price=id_{\Reals},\]
is a risk measurement regime on $\Linfty$ such that $\Rho(X):=\sup_{k\in\Nat}\Erw_{\Quot_k}[X]$, $X\in\Linfty$, is a coherent monetary risk measure which is continuous from above and sensitive with respect to the strong reference model $\Prob$. We consider $X:=id_{\mathbb Z}$. We first observe that for all $k\in\Nat$ it holds that $\Erw_{\Quot_k}[|X|]=1$, which is sufficient for $X\in\Lrho$. Using the notational conventions of Theorem \ref{mainresult}, for all $n\in\Nat$
\[\tilrho(X_n)\geq \Rho(X_n^{n^2})\geq \Erw_{\Quot_{n^2}}[X_n^{n^2}]=\frac 1 2\left(1-\frac 1 n\right).\]
Hence $\eta(X)\geq\frac 12$. 
However, for $m\in\Nat$ fixed, we obtain for all $n>m$ that 
\[\Erw_{\Quot_k}[X^m_n]=\begin{cases}0\tn{, if }k\leq m,\\
\frac m{2k}-\frac 1 2\tn{, if }m<k\leq n,\\ \frac{m-n}{2k}\tn{, if }k>n.\end{cases}\]
This implies 
\[\tilrho(X^m)=\xi(X^m)=\lim_{n\to\infty}\Rho(X^m_n)=0,\]
and therefore
\[\tilrho(X)=\lim_{m\to\infty}\tilrho(X^m)=0<\frac 1 2\leq\eta(X).\]
\end{example}

\begin{example}[$\Mrho\subsetneq \heart\subsetneq \Lrho$]\label{inclusions}
Let $(\Omega,\Fcal)$ be the real numbers endowed with their Borel sets $\Borel(\Reals)$. Let $\Prob_0$ be the probability measure $\Prob$ from Example \ref{ex:etaneqtilrho} extended to $\Borel(\Reals)$, and define $\Prob_1$ by its Lebesgue density $d\Prob_1=e^{1-x}\Ind_{(1,\infty)}dx$. Let $\Prob:=\frac 1 2(\Prob_0+\Prob_1)$, and consider the risk measurement regime
\begin{center}$\acc:=\{X\in\Linfty\mid \forall k\in\Nat:~\Erw_{\Quot_k}[X]\leq 0\tn{, and }\Erw_{\Prob_1}[e^X]\leq 1\},\quad \Scal=\Reals,\quad \price=id_{\Reals},$\end{center}
where the probability measures $(\Quot_k)_{k\in\Nat}$ are chosen as in Example \ref{ex:etaneqtilrho} and extended to $\Borel(\Reals)$. One can easily show that $\Rho(X)=\rho_0(X)\vee\rho_1(X)$, where 
\begin{center}
$\rho_0(V)=\sup_{k\in\Nat}\Erw_{\Quot_k}[V],\quad\rho_1(U)=\log\left(\Erw_{\Prob_1}[e^U]\right),\quad U,V\in\Linfty$.\end{center}
$\Rho$ is a sensitive finite risk measure on $\Linfty$ being continuous from above, and $\Prob$ is a strong reference probability model. \\
Consider first $X\in L^0_{\Prob}$ be generated by $id_\mathbb Z$. We have already shown in Example \ref{ex:etaneqtilrho} that $\rho(t|X|)=\rho_0(t|X|)=t$ for all $t\geq 0$, hence $X\in\heart$. Nevertheless, it holds for all $k\in\Nat$ that 
\begin{center}$1\geq \rho(|X|\Ind_{\{|X|>k\}})\geq\Erw_{\Quot_{k+1}}[|X|\Ind_{\{|X|>k\}}]=1$.\end{center}
Hence $\lim_k\rho(|X|\Ind_{\{|X|\geq k\}})= 1$, which is sufficient for $X\in\heart\backslash\Mrho$ by Corollary~\ref{Mrhocharac}.\\
Let now $\lambda>0$ and define $Y\in L^0_\Prob$ generated by 
\[\omega\mapsto\begin{cases}\frac 1{\lambda}(\omega-1),~\om\in (1,\infty)\backslash\mathbb Z,\\0\tn{ otherwise.}\end{cases}\]
$Y$ is exponentially distributed under $\Prob_1$ with parameter $\lambda$. Moreover, $Y\in\Lrho$ and satisfies $\tilrho(Y)<\infty$. However, for every $t>\lambda$, $\rho(t|Y|)=\log(\Erw_{\Prob_1}[e^{tY}])=\infty$, hence $Y\in C^{\mathcal R}$.
\end{example}

\begin{example}[Entropic Risk Measure]\label{ex:entropic}
On a probability space $(\Omega,\Fcal,\Prob)$ consider for $\beta>0$ fixed the \textit{entropic risk measure} $\Rho(X):=\frac 1 {\beta}\log\left(\Erw_\Prob[e^{\beta X}]\right)$, $X\in\Linfty$. One can easily show 
\begin{center}$\Lrho=\{X\in L^0_\Prob\mid\exists k>0: e^{k|X|}\in L^1_\Prob\},\quad\tilrho(X)=\frac 1 {\beta}\log\left(\Erw_\Prob[e^{\beta X}]\right),\quad X\in\Lrho$.\end{center}
$\tilrho$ is tail continuous. Indeed, choose $X\in\dom(\tilrho)$ arbitrary and $Y\in\Lrho$ such that $X+Y^+\in\dom(\tilrho)$, i.e. $e^{\beta X+\beta Y^+}\in L^1_\Prob$. By continuity of $\log$ and dominated convergence, we obtain
\[\lim_r\tilrho(X+Y\Ind_{\{Y\geq r\}})=\lim_r\frac 1 {\beta}\log\left(\Erw_\Prob[e^{\beta (X+Y^+)}\Ind_{\{Y\geq r\}}]+\Erw_\Prob[e^{\beta X}\Ind_{\{Y<r\}}]\right)=\tilrho(X).\] 
\end{example}

\begin{example}[AVaR-based risk measures]\label{ex:avar}
Consider the Average Value at Risk $AVaR_\alpha$ for some $\alpha\in (0,1]$ on $\Linfty$, which is known to have the minimal dual representation
$$AVaR_\alpha(X)=\max_{\Quot\in {\mathcal Q}_\alpha}\Erw_\Quot[X], \quad X\in \Linfty,$$ where $${\mathcal Q}_\alpha:=\left\{\Quot\ll \Prob\mid \frac{d\Quot}{d\Prob}\leq \frac{1}{1-\alpha}\right\},$$ see \cite[Theorem~4.52]{5}. Given the  acceptance set $\acc:=\{X\in \Linfty\mid AVaR_\alpha(X)\leq 0\}$ we define the risk measurement regime $\mathcal{R}=(\acc, \Scal, \price)$ by $\Scal=\Reals\cdot U$ for some $U\in \Linfty$ with $\Prob(U>0)=1$, and $\price(mU):=m$, $m\in \Reals$. By \cite[Proposition 4.4]{FKM2013} and Proposition~\ref{ancillary} the resulting risk measure $\Rho$ is finite, continuous from above with strong reference model $\Prob$, and $\Lrho=\heart=\Mrho$. 
Hence, Theorem~\ref{thm:tilrho:eta} applies for all $X\in \Lrho$.   
\end{example}

\begin{appendix}

\section{Proofs of Proposition \ref{ancillary} and Corollary \ref{cor:ancillary}}\label{appendix:A}

\begin{proof}[Proof of Proposition \ref{ancillary}](i): Suppose $\Rho(0)$ is a real number and let $X=Y+N$ for some $Y\in\acc$ and some $N\in\ker(\price)$. By $\Scal$-additivity, $\Rho(X)=\Rho(Y)\leq 0$ holds. $\Rho$ being l.s.c. implies $\Rho(X)\leq 0$ for all $X\in\cl_{|\cdot|_{\infty}}(\acc+\ker(\price))$. Again, $\Scal$-additivity of $\Rho$ allows to infer $\price(Z)=\Rho(Z)-\Rho(0)\leq -\Rho(0)$ for all $Z\in\cl_{|\cdot|_{\infty}}(\acc+\ker(\price))\cap\Scal$, and $\price$ is bounded from above on the latter set. By virtue of \cite[Theorems 2 and 3]{FKM2015}, $\mathcal B(\acc)\cap \Ecal_{\price}\neq\emptyset$ and 
\[\Rho(X)=\sup_{\mu\in\mathcal B(\acc)\cap\Ecal_{\price}}\int X\,d\mu-\sigma_{\acc}(\mu),\quad X\in\Bd,\]
from which the claimed equations \eqref{support} and \eqref{eq:dualrep1} are derived easily. 
If $\Rho$ is coherent, its positive homogeneity implies that $\Rho^*|_{\dom(\Rho^*)}\equiv 0$. As furthermore $\dom(\Rho^*)=\mathcal B(\acc)\cap\mathcal E_\price$, \eqref{support:coh} and \eqref{eq:dualrep:coh} are special cases of \eqref{support} and \eqref{eq:dualrep1}.

\smallskip\noindent
(ii): $\Rho^*$ is by definition a $\sigma(\bounded, \Bd)$-l.s.c. function, hence its lower level sets are closed in this topology. Let $c\in\Reals$ and suppose $\mu\in E_c$, thus \textit{a fortiori} $\mu\in \bounded_+$. \eqref{eq:dualrep1} implies 
\begin{center}$\forall\mu\in E_c:\quad\mu(\Omega)\leq\Rho^*(\mu)+\Rho(1)\leq c+\Rho(1)<\infty.$\end{center}
Thus being a closed subset of a dilation of the closed unit ball of $\bounded$, $E_c$ is weakly* compact by virtue of the Banach-Alaoglu Theorem \cite[Theorem 6.25]{Ali}.

\smallskip\noindent
(iii): Assume first a risk measure $\Rho$ associated to the acceptance set $\acc$ is finite and continuous from above. $\Rho$ is in particular norm-continuous by Remark \ref{shapiro}(v) and statements (i) and (ii) apply. Continuity from above implies $\Rho(k\Ind_{A_n})\downarrow \Rho(0)$ for all $k>0$ whenever $(A_n)_{n\in\Nat}\subseteq \mathcal F$ is a sequence of events decreasing to $\emptyset$. For $\mu\in\dom(\Rho^*)$, 
\[-\Rho(0)\leq \sup_{k>0}k\lim_{n\to\infty}\mu(A_n)-\Rho(0)=\sup_{k>0}\lim_{n\to\infty}k\mu(A_n)-\Rho(k\Ind_{A_n})\leq\Rho^*(\mu)<\infty.\]
This can only hold if $\lim_n\mu(A_n)=0$, i.e. $\mu\in\countable_+$. By (ii), this is equivalent to all level sets $E_c$ of $\Rho^*$, $c\in\Reals$, being $\sigma(\countable,\Bd)$-compact. \\
For the converse, assume that $\dom(\Rho^\ast)\subseteq \countable_+$. Let $(X_n)_{n\in\Nat}$ be any sequence in $\Bd$ such that $X_n\downarrow X$ for some $X\in\Bd$. Let $Y\in\{X,X_1, X_2, ...\}$ and suppose that $\mu \in\dom(\Rho^*)$ satisfies 
$\Rho(Y)-1\leq \int Y\,d\mu-\Rho^*(\mu).$
We can thus use the monotonicity of $\Rho$ and the positivity of $\mu$ to estimate 
\begin{align*}\Rho^*(\mu)&\leq \int Y\,d\mu-\Rho(Y)+1\leq\int X_1d\mu-\Rho(X)+1\\
&\leq \frac 1 2\Rho(2X_1)+\frac 1 2 \Rho^*(\mu)-\Rho(X)+1.
\end{align*}
Rearranging this inequality yields that 
\[\Rho^*(\mu)\leq 2+\Rho(2X_1)-2\Rho(X)=:c,\]
a bound which is independent of $Y$. Therefore, for all $Y\in\{X,X_1,X_2,...\}$ it holds that
\[\Rho(Y)=\sup_{\mu\in E_c}\int Y\,d\mu-\Rho^*(\mu)=\max_{\mu\in E_c}\int Y\,d\mu-\Rho^*(\mu),\]
where in the last equality we used the $\sigma(\countable,\Bd)$-continuity of $\mu\mapsto\int Y\,d\mu-\Rho^*(\mu)$ and the compactness of $E_c$. For each $n\in\Nat$ choose $\mu_n\in E_c$ such that 
$$\Rho(X_n)=\int X_n\,d\mu_n-\Rho^*(\mu_n).$$
Note that $E_c$ is $\sigma(\countable, \countable^\ast)$-compact by virtue of \cite[Theorem 4.7.25]{13}. The Eberlein-Smulian Theorem (see e.g.\ \cite[Theorem~6.38]{Ali}) now implies that we may select a $\sigma(\countable,\countable^\ast)$-convergent subsequence $(\mu_{n_k})_{k\in\Nat}$ with limit $\bar\mu\in E_c$. Choose a measure $\nu$, for instance $$\nu:=\bar \mu+ \sum_{k\in \N}\frac{1}{2^k}\mu_{n_k},$$ such that for all $\mu\in{\cal K}:=\{\bar \mu, \mu_{n_1}, \mu_{n_2}, \ldots \}$ we have $\mu\ll \nu $. As $\nu(A)\leq\eps$ implies $\mu(A)\leq \eps$ for all $A\in\Fcal$ and the set of Radon-Nikodym derivatives $\{\frac{d\mu}{d\nu}\mid\mu\in {\cal K}\}$ is $\Norm_{L^1_\nu}$-bounded as a subset of $L^1(\nu)$, we conclude that they form a $\nu$-uniformly integrable family by \cite[Proposition 4.5.3]{13}. Abbreviating $Z_k:=\frac{d\mu_{n_k}}{d\nu}$, we obtain for all constants $L>0$
\begin{align*}\limsup_{k\to \infty} \left|\int X\,d\bar\mu-\int X_{n_k}\,d\mu_{n_k}\right|&\leq\limsup_{k\to \infty} \left|\int X\,d\bar\mu-\int X\,d\mu_{n_k}\right|+\left|\int (X-X_{n_k})\,d\mu_{n_k}\right|\\
&\leq \limsup_{k\to \infty }\int_{\{Z_k\geq L\}}|X_1-X|Z_k\,d\nu+\int_{\{Z_k<L\}}|X_{n_k}-X|L\,d\nu\\
&= \limsup_{k\to \infty}\int_{\{Z_k\geq L\}}|X_1-X|Z_k\,d\nu=0,
\end{align*}
where we applied monotone convergence for the second but last equality and where the last equality follows from the uniform $\nu$-integrability of the densities $Z_k$ and the fact that $|X_1-X|$ is bounded by a constant. Hence $\lim_k\int X_{n_k}\,d\mu_{n_k}=\int X\,d\bar\mu$, and from lower semicontinuity of $\Rho^*$, we arrive at 
\begin{align*}\lim_{k\to\infty}\Rho(X_{n_k})=\limsup_{k\to\infty}\int X_{n_k}\,d\mu_{n_k}-\Rho^*(\mu_{n_k})\leq\int X\,d\bar\mu-\Rho^*(\bar\mu)\leq\Rho(X).\end{align*}
$\Rho(X)\leq\inf_{n\in\Nat}\Rho(X_n)=\lim_{k\to\infty}\Rho(X_{n_k})$ holds \textit{a priori}, however. We infer $\Rho(X)=\lim_n\Rho(X_n)$.\\
Suppose $\mathcal B(\acc)\subseteq \countable_+$. By \eqref{support} and statement (ii), the lower level sets of the dual conjugate of any finite risk measure $\Rho$ associated to $\acc$ are $\sigma(\countable,\Bd)$-compact, and continuity from above follows from the equivalence proved just before.

\smallskip
\noindent (iv): Suppose that the risk measurement regime $\mathcal R=(\acc,\Scal,\price)$ is such that $\Scal=\Reals\cdot U$ for some $U\in\Bd_{++}$ and such that the resulting risk measure is finite. Assume for contradiction the existence of a $0\neq \mu\in\mathcal B(\acc)$ such that $\int U\,d\mu=0$. Recall that $\mu$ is necessarily positive, and let $k>\frac{\sigma_{\acc}(\mu)}{\mu(\Omega)}$. For any $r\in\Reals$
$$\int(k-r U)d\mu=k\mu(\Omega)>\sigma_\acc(\mu),$$
which would imply that $k-rU\notin\acc$ for any $r\in\Reals$, and thus $\Rho(k)=\infty$ in \textsc{contradiction} to the finiteness of $\Rho$. As hence $\int U\,d\mu>0$ has to hold for all $0\neq \mu\in\mathcal B(\acc)$, we can identify with \eqref{support} 
$$\dom(\Rho^*)=\mathcal E_\price\cap\mathcal B(\acc)=\left\{\frac{\price(U)}{\int U\,d\mu}\mu\Big| 0\neq\mu\in\mathcal B(\acc)\right\},$$ and by \eqref{eq:dualrep1}, 
$$\Rho(X)=\sup_{0\neq\mu}\frac{\price(U)}{\int U\,d\mu}\left(\int X\,d\mu-\sigma_{\acc}(\mu)\right),\quad X\in\Bd,$$
is the minimal dual representation of $\Rho$. From this representation and (iii), we infer that $\Rho$ is continuous from above if and only if $\mathcal B(\acc)\subseteq\countable$.
\end{proof}

\begin{proof}[Proof of Corollary \ref{cor:ancillary}]
As $\Rho(X)\leq 0$ for all $X\in\acc+\ker(\price)$, finiteness and $\Scal$-additivity of $\Rho$ together with \cite[Remark 6]{FKM2015} show that $\acc+\ker(\price)$ is proper and thus an acceptance set. Fix $U\in\Scal\cap\Bd_{++}$ and recall from \cite[Lemma 3]{FKM2015} the identity
$$\Rho(X)=\inf\{\price(rU)\mid r\in\Reals, X-rU\in\acc+\ker(\price)\},\quad X\in \Bd.$$
\textit{A fortiori}, $\mathcal R':=(\acc+\ker(\price),\Reals\cdot U,\price|_{\Reals\cdot U})$ is a risk measurement regime and the associated risk measure $\rho_{\mathcal R'}$ is continuous from above if and only if $\Rho$ is continuous from above. As the identity $\mathcal B(\acc+\ker(\price))=\mathcal B(\acc)\cap\ker(\price)^\perp$ is easily verified, the claimed equivalence follows from Proposition \ref{ancillary}(iv). 
\end{proof}

\section{Proof of Theorem \ref{R1111}}\label{appendix:B}

The proof heavily relies on the following result. 

\begin{lemma}[Grothendieck; see Exercise 1, Chapter 5, part 4 of \cite{12}]\label{R11} A convex subset $C$ of $\Linfty$ is closed in the $\sigma(\Linfty,\countable_\Prob)$-topology if and only if for arbitrary $r>0$ the set $C_r:=\{X\in C\mid\|X\|_{\infty}\leq r\}$
is closed with respect to convergence in probability, i.e.\ with respect to the metric \[d_{\Prob}(X,Y):=\Erw_\Prob[|X-Y|\wedge 1].\] 
\end{lemma}

\noindent(i): For a set $\Gamma\subseteq \Linfty$ we define $\Gamma^{\diamond}:=\{\mu\in \countable_\Prob\mid\forall X\in\Gamma:\,\int X\,d\mu\leq 1\}$, the one-sided polar of $\Gamma$. Moreover, $\Gamma^{\diamond}=(\cl_{\sigma(\Linfty,\countable_\Prob)}(\Gamma))^{\diamond}$. Having this information at hand, one can easily identify\footnote{ The only difficult part is the following: Recall from Remark \ref{shapiro}(iv) that $\Rho(Y)\leq 0$ if and only if $Y\in\cl_{\Norm_{\infty}}(\acc+\ker(\price))$. Assume $\nu\in\mathbf{C}^{\diamond}$, then $\nu\in(\countable_\Prob)_+$ by normalisation and monotonicity of $\acc$. Also, $\mathbf C$ being a cone shows
$$\cenv^\diamond=\{\nu\in\countable_\Prob\mid\forall Y\in\cenv:~\int Y\,d\nu\leq 0\}.$$ Let $U\in\Scal\cap(\Linfty)_{++}$ and $X\in\Linfty$. Since $\Rho(X-\frac{\Rho(X)}{\price(U)}U)=0$, we obtain that either $c:=\int\frac 1 {\price(U)}U\,d\nu=0$, which implies $\nu(\Omega)=0$ and $\nu=0$, or 
$$c\sup_{X\in\Linfty}\left(\int X\,d\left(\frac{\nu}c\right)-\Rho(X)\right)\leq 0\quad\Longrightarrow\quad\frac{\nu}c\in E_0.$$} 
\begin{center}$\{c\mu\mid c\geq 0, \mu\in E_0\}=\left(\bigcup\{t\acc+\ker(\price)\mid t\geq 0\}\right)^{\diamond}=\cenv^{\diamond}.$\end{center}
From the Bipolar Theorem \cite[Theorem 5.91]{Ali} we deduce $$\cenv=\{X\in\Linfty\mid \forall \mu\in E_0: \int X\,d\mu\leq 0\}.$$ Hence $\cenv$ is an acceptance set. Consider the following risk measurement regime and its implied risk measure:
\[\check{\cal R}:=\left(\cenv,\Scal,\price\right),\quad \rho_{\check{\cal R}}(X)=\sup_{\mu\in E_0}\int X\,d\mu,~ X\in\Linfty.\]
$\rho_{\check{\cal R}}$ is finite, coherent, and continuous from above by Lemma~\ref{lem:strongref}. Hence, by Lemma~\ref{lem:cohpcal}, $\mathcal P\neq\emptyset$ is equivalent to $\rho_{\check{\cal R}}$ being sensitive, i.e. $\cenv$ does not contain any element in $(\Linfty)_{++}$. 

\smallskip\noindent
(ii): Suppose that $\cenv\cap(\Linfty)_{++}$ is non-empty. Using the monotonicity and conicity of $\cenv$, we can find some $B\in {\cal F}_+$ such that $\Ind_B\in \cenv$.
Let us define the sets
\begin{center}$\mathbf{D} :=\{Y=d_{\Prob}\tn-\lim_nt_nW_n\mid t_n\geq 0, W_n\in\acc+\ker(\price)\}$,\\
$\mathbf{D}_r=\{Y\in \mathbf D\mid\|Y\|_{\infty}\leq r\},~ r>0.$\end{center}
$\mathbf D$ is a convex cone. It is straightforward to check that $\mathbf{D}_r$ is $d_\Prob$-closed. We apply Grothendieck's Lemma \ref{R11} to infer that $\mathbf{D}$ is a $\sigma(\Linfty, \countable_\Prob)$-closed cone. Thus the inclusion $\cenv\subseteq \mathbf{D}$ holds and we must be able to find sequences $(X_n)_{n\in\Nat}\subseteq\acc$, $(Z_n)_{n\in\Nat}\subseteq\ker(\price)$ and $(t_n)_{n\in\Nat}\subseteq(0,\infty)$ such that $\Ind_B=d_{\Prob}$-$\lim_nt_n(X_n+Z_n)$. Define $V_n:=t_n(X_n+Z_n)$, $n\in\Nat$. 
Without loss of generality we can assume that $\|X_n+Z_n\|_{\infty}\leq 1$. Otherwise, note that by normalisation $0\in\mathbb A:=\cl_{\Norm_{\infty}}(\acc+\ker(\price))$, and $\cenv$ is also the smallest $\sigma(\Linfty,\countable_\Prob)$-closed cone that contains $\mathbb A$; thus we can shift to 
\[\frac{X_n+Z_n}{\|X_n+Z_n\|_{\infty}}\in\mathbb A,\quad \tilde t_n=\|X_n+Z_n\|_{\infty}t_n.\]
As $\mathbb A$ is convex, $(t_n)_{n\in\Nat}$ cannot be bounded. If there is some $M>0$ such that $\sup_{n\in\Nat}t_n\leq M$,
we can define the sequence $(t_n(X_n+Z_n)/2M)_n\subseteq\mathbb A$ which converges in probability and with respect to $\sigma(\Linfty,\countable_\Prob)$ to $\Ind_B/2M\in \mathbb A$. As $\mathbb A=\{Y\mid \Rho(Y)\leq 0\}$, we would obtain a \textsc{contradiction} to the sensitivity of $\Rho$ and can therefore assume $t_n\uparrow\infty$. $d_{\Prob}(V_n,\Ind_B)\to 0$ for $n\to\infty$ implies
\begin{center}$V_n^-\overset{d_{\Prob}}{\longrightarrow}0,\quad V_n^+\overset{d_{\Prob}}{\longrightarrow}\Ind_B,\quad n\to\infty,$\end{center}
and this means that
\[\limsup_{n\rightarrow\infty}\Prob\left(\left\{V_n^+\geq\frac 1 2\right\}\cap B\right)=\Prob(B),\]
The rule of equal speed of convergence is violated.

\smallskip\noindent
(iii): Let $(X_n)_{n\in\Nat}\subseteq\acc,(Z_n)_{n\in\Nat}\subseteq\ker(\price)$ and $(t_n)_{n\in\Nat}$ be sequences violating the rule of equal speed of convergence such that the rescaled sequence $(V_n)_{n\in\Nat}$ is bounded in the $\Norm_{\infty}$-norm. Let $B$ be a measurable set with positive probability such that
\[\limsup_{n\rightarrow\infty}\Prob(\{V_n\geq \eps\}\cap B)=\Prob(B).\]
Let $\mu\approx\Prob$ be a finite measure with $\int Z\,d\mu=\price(Z)$ for all $Z\in\Scal$, and let $\eta>0$ be an arbitrary positive number. 
Note that due to the Dominated Convergence Theorem and the bounded vanishing in probability of $V_n^-$, we obtain for $n\to\infty$ the behaviour $\lim_n\int V_n\Ind_{\{V_n\leq-\eta\}}\,d\mu\rightarrow 0.$
For all $n$ large enough such that $|\int V_n\Ind_{\{V_n\leq-\eta\}}\,d\mu|<\eta$ we can estimate 
\begin{align*}\int V_n\,d\mu&\geq\eps\mu(\{V_n\geq \eps\}\cap B)-\eta \mu(V_n\in(-\eta,\eps))-\eta.\end{align*}
Thus for all $\eta>0$ our assumption yields the estimate $$\limsup_{n\rightarrow\infty}\int V_n\,d\mu\geq\eps \mu(B)-\eta(\mu(\Omega)+1).$$
Sending $\eta\downarrow 0$, we obtain from $\mu\approx\Prob$ that $\limsup_{n\rightarrow\infty}\int V_n\,d\mu\geq \eps \mu(B)>0.$
After choosing $n$ suitably we have found a vector in $\acc+\ker(\price)$ such that $\int(X_n+Z_n)\,d\mu>0$, hence $\mu\notin E_0$. 

\end{appendix}

\bigskip

\noindent\textbf{Acknowledgements:} We would like to thank an anonymous referee for helpful comments that aided in improving the original draft of this manuscript.

\bibliographystyle{plain}

\end{document}